\documentclass[final,letterpaper,journal,twocolumn]{IEEEtran}
\usepackage{amsmath}
\usepackage{amsfonts}
\usepackage{graphicx}
\usepackage{cite}
\usepackage{color}

\newtheorem{theorem}{Theorem}

\newtheorem{definition}[theorem]{Definition}

\newtheorem{lemma}[theorem]{Lemma}

\newtheorem{proposition}[theorem]{Proposition}
\newtheorem{remark}[theorem]{Remark}

\usepackage{fancyhdr}

\begin{document}

\title{Energy Efficient Resource Allocation\ for Time-Varying OFDMA Relay
Systems with Hybrid Energy Supplies}
\author{Bo Yang, Yanyan Shen, Qiaoni Han, Cailian Chen, Xinping Guan, and
Weidong Zhang
\thanks{Copyright (c) 2015 IEEE. Personal use of this material is permitted. However, permission to use this material for any other purposes must be obtained from the IEEE by sending a request to pubs-permissions@ieee.org.}
\thanks{%
B. Yang, Q. Han, C. Chen, X. Guan and W. Zhang are with Shanghai Jiao Tong
University, Shanghai, 200240, P.R. China. Y. Shen is with Shenzhen
Institutes of Advanced Technology, Chinese Academy of Sciences, Shenzhen,
Guangdong Province 518000, P.R. China.}}
\maketitle
\thispagestyle{fancy}
\fancyhead[L]{DOI: 10.1109/JSYST.2016.2551319}
\cfoot{}

\begin{abstract}
This paper investigates the energy efficient\ resource allocation for
orthogonal frequency division multiple access (OFDMA) relay systems, where
the system is supplied by the conventional utility grid and a renewable
energy generator equipped with a storage device. The optimal usage of radio
resource depends on the characteristics of the renewable energy generation
and the mobile traffic, which exhibit both temporal and spatial diversities.
Lyapunov optimization method is used to decompose the problem into the joint
flow control, radio resource allocation and energy management without
knowing a priori knowledge of system statistics. It is proven that\emph{\ }%
the proposed algorithm can result in close-to-optimal performance with
capacity limited data\ buffer and storage device. Simulation results show
that the flexible tradeoff between the system utility and the conventional\
energy consumption can be achieved. Compared with other schemes, the
proposed algorithm demonstrates better performance.
\end{abstract}


\section{Introduction}

\setcounter{page}{1} Orthogonal frequency division multiple access (OFDMA)
is regarded as a promising technology in the future wireless communications,
due to its advantages of high spectral efficiency and resistance to
multipath fading.\ On the other hand, cooperative relaying is able to
enhance the transmission range, system capacity and reliability. Thus, the
relay-assisted OFDMA system has been\ attracting intensive research
interests in both academia and industry since it can provide cooperative
diversity, channel diversity and multiuser diversity gains at the same time.

Although the transmission quality of wireless systems can be improved with
cooperative communications, the new challenges due to the dramatic increase
of power cost and traffic demands in wireless communications\ have been
increasingly prominent. About half of the mobile service provider's annual
operating expenses are power costs and base stations (BSs) consume over\
80\% of the total energy of the network as claimed by\ Fehske \textit{et al.}%
\ \cite{Fehske80} and Feng \textit{et al.} \cite{Geoffrey_Li}. Moreover the
largest fraction of greenhouse gas emissions caused by radio networks come
from access networks, and mostly from BSs. Then the needs for energy saving
of radio networks advocate green radio communications. Green radio
communications will not only help mobile operators to enhance the energy
efficiency of BSs but also reduce the emissions of CO$_{2}$ by using
renewable energy\ sources (\textit{e.g.,} wind and solar power) to supply
wireless communications in addition to the constant alternating current (AC)
from utility grid, e.g., the works by Ho \textit{et al.} \cite{RuiZhang}
and\ Niyato \textit{et al}. \cite{duisit_niyato}. As to the relay-assisted
transmission powered by renewable energy source, like the idea of demand
response in \cite{Mao_shiwen}-\cite{Jiming_chen}\ the relay selection and
radio resource allocation should be adaptive to the variations in harvested
renewable energy in order to ensure the quality of service (QoS) of mobile
users with the minimum AC energy. In addition, the joint design of relay
selection, subcarrier allocation and power control in relay-assisted OFDMA
networks\ is a challenge since all the considered variables are coupled
together, let alone the mentioned time-varying renewable energy generation.

In this paper, we consider a downlink OFDMA relay-assisted network
consisting of multiple relay stations and a BS, which is powered by hybrid
energy, \textit{i.e.,} renewable energy and constant AC from the grid.
Although the renewable energy provided by harvester is free, it is
intermittent over time due to time-varying solar and wind patterns.\ In
order to smooth the uncertain supply of renewable source, an energy storage
device is equipped with the BS to buffer the intermittent renewable energy
generation. Since the capacity of energy storage device and data buffer at
the BS\ are finite, it is necessary to propose\ a joint design in order to
coordinate the downlink transmission, radio resource allocation and energy
management of the storage device in order to maximize the throughput
utility\ with the least grid energy in a long-term manner.\ In summary, the
contributions of this paper are as follows:

\begin{itemize}
\item From the perspective of network operator, we formulate a stochastic
optimization problem to keep a balance between the throughput utility and
on-grid energy consumption, which are usually two conflicting objectives.

\item An online algorithm \emph{FREE} is proposed to regulate the downlink
data requests, allocate radio resource and manage\ energy in the\ storage
device under the constraints of finite buffer size and storage capacity.

\item The proposed \emph{FREE} can arbitrarily approach the optimal tradeoff
between system throughput and on-grid\ energy consumption\ by theoretic
analysis and simulation results.
\end{itemize}

The rest of this paper is organized as follows. Section II presents a review
of related works. Section III describes the system model and the components
of the system. Section IV casts the problem into a stochastic optimization
problem. Section V presents an online algorithm to jointly allocate resource
and manage energy. Section VI analyzes the performance of the\ proposed
algorithm. Simulations are provided in Section VII. The paper is concluded
in Section VIII.

\section{Related Works}

Since this paper considers the scenario of relay-based cooperative networks
with hybrid energy supply, the related works on cooperative communications
are first reviewed and then some latest results on renewable energy powered
communications are discussed.

Various cooperative communication schemes such as Amplify-and-Forward (AF)
and Decode-and-Forward (DF) have been proposed in the literature \cite%
{Feifei_Gao}-\cite{SongGuo}. Compared with\ a DF relay decoding,
re-modulating and retransmitting the received signal, an AF one simply
amplifies and retransmits the signal towards destinations without decoding.
Wang \textit{et al. }in\ \cite{WangDi} and Shen \textit{et al}. in \cite%
{YanyanShen} proposed power control for fairness in AF relay networks. The
power control of multiple users\ with AF relays are considered by Zappone
\textit{et al}. in\ \cite{Zappone_relay_game}\ and by Cheung \textit{et al}.
in \cite{southampton} for energy-efficiency. Depending on the channel
condition and performance criterion, DF may outperform AF, or vice versa.
For DF relay strategy,\ \cite{tao_wang}-\cite{Shenyanyan_singapore}
considered the resource allocation by solving a class of static optimization
problem. The joint design of congestion control and relay resource
allocation for LTE-A was proposed in \cite{Benliang_TWC} by solving a
stochastic optimization problem. Tang \textit{et al}. in\ \cite{Tangjian}
proposed three greedy algorithms for multiuser multiple relay networks to
maximize the network utility and comprehensive theoretical analysis was
given to measure the worst performance. Different from the assumption of
half-duplex at relay, Li \textit{et al}. in \cite{SongGuo} proposed a
full-duplex cooperative communication scheme.\ In this paper, the
energy-aware\ relay-assisted transmission aims to balance the throughput
utility and on-grid energy consumption. Since the transmission power of
relays is coupled with the power allocation of the BS, which is supplied by
hybrid energy, the relay power allocation, transmission mode selection
(cooperative transmission or direct transmission) and subcarrier allocation
are all affected by queue states, residual energy in storage device. This
makes the radio resource allocation problem significantly different from the
aforementioned works.

To reduce the consumption of on-grid energy in wireless access networks,
there are some works considering that the transmitter is powered by hybrid
energy source \cite{Ahmed2013}-\cite{Peter}. This paper considers that the
transmitter is also powered by hybrid energy source in a relay-assisted
cooperative communication scenario, where an online joint resource
allocation and energy management algorithm is proposed with lower
computational complexity compared with the dynamic programming based ones in
\cite{Ahmed2013}-\cite{ng}. Compared against \cite{Peter}\ the salient
feature of the proposed algorithms in this paper is that its implementation
is independent of the energy harvesting profiles.\ For the delay/throughput
optimal transmission with renewable energy sources, the schemes on\
transmission power control and energy management in battery for
point-to-point transmissions were proposed in \cite{Aylin_Yener}-\cite%
{HuangHUang}. The network planning problem with sustainable energy sources
in OFDMA relay networks was considered by Zheng \textit{et al.} in \cite%
{Shen_relay_renew} to ensure network connectivity and users' QoS
requirements.\ To further reduce the circuit energy consumption of cellular
networks, the BSs or relay stations (RSs) can be switched into sleep mode
dynamically\ by utilizing the spatial and temporal fluctuations of traffic
loads and inter-cell cooperation \cite{Meo}-\cite{Ansari}. The most related
work also includes \cite{Ahmed}, where Ahmed \textit{et al.} consider power
control for a cooperative system consisting of one energy harvesting (EH)
source and one EH relay communicating to one destination. Almost all the
above works assume that either the data buffer or the storage device is
infinite to make their problem tractable. Also the prior knowledge on random
events is needed in their designs. In this paper, the trade-off between
throughput utility and grid energy consumption is considered for
multi-user/multi-relay networks under the constraints of finite data and
energy buffer. An online algorithm is designed without relying on any
statistical information.

\section{System Model}

\subsection{Physical Layer Model}

We consider the downlink transmission in a two-hop relay-assisted network,
which consists of a set of users $\mathcal{N=}\left\{ 1,2,\cdots ,N\right\} $
indexed by $n,$ a set of RSs $\mathcal{R=}\left\{ 1,2,\cdots ,K\right\} $
indexed by $i$ and one BS indexed by $B$.$\ $OFDMA is employed over both
hops and the whole bandwidth $W$ is equally divided into $M$ subcarriers
with sub-bandwidth $b=\frac{W}{M}.$ The set of subcarriers is $\mathcal{M=}%
\left\{ 1,2,\cdots ,M\right\} $ and the subcarrier is indexed by $m$. The
noise at each receiver on each subcarrier is zero-mean circular symmetric
complex Gaussian with variance $bN_{0}.$\ The system is assumed to be a
time-slotted one with equal time interval normalized to be unit time.
Hereafter $b$ is assumed to be one without loss of generality.

The BS can transmit directly to each user or\ employ one of the RSs for
cooperative communications. If the cooperative communication is adopted,
each transmission consists of two stages. In the first stage, the BS
communicates with RSs and users on some subcarriers. In the second stage,
the selected RS on the same subcarriers helps the BS transmit to the user
with decode-and-forward (DF) strategy\footnote{%
The performance of DF outperforms AF depending on channel conditions and
performance criterion. In this paper, we only consider DF strategy. The
similar technique can be applied to AF relay.} with the same codebook as
that of the BS. Let $\{\alpha _{i,n}^{m}:\alpha _{i,n}^{m}\in \left\{
0,1\right\} ,\forall n\in \mathcal{N},i\in \mathcal{R},m\in \mathcal{M}\}$
denote the set of RS, user and subcarrier matching indicator, where $\alpha
_{i,n}^{m}=1$ indicates that the BS can communicate with user $n$ via relay $%
i$ over the $m^{\text{th}}$ subcarrier, and $\alpha _{i,n}^{m}=0$ otherwise.
Note that $\alpha _{B,n}^{m}=1$ means that the BS will communicate to user $%
n $ over subcarrier $m$ directly and $\alpha _{B,n}^{m}=0$ otherwise. An
exclusive subcarrier assignment is enforced such that at most one user or
one user-relay pair can be allocated to a single subcarrier. Also, on one
subcarrier, only one transmission mode (cooperative or direct) is allowed.%
\textit{\ }These constraints are as follows,$\forall n\in \mathcal{N},i\in
\mathcal{R},m\in \mathcal{M}$
\begin{equation}
\sum_{n\in \mathcal{N}}\sum_{i\in \mathcal{R}}\left( \alpha
_{i,n}^{m}+\alpha _{B,n}^{m}\right) \leq 1,\text{ }\alpha _{i,n}^{m}\in
\left\{ 0,1\right\} ,\alpha _{B,n}^{m}\in \left\{ 0,1\right\} .  \label{2}
\end{equation}

We use $c_{B,n}^{m}$ to denote the achievable rate by the BS transmitting to
user $n$ directly over the $m^{\text{th}}$ subcarrier. We have
\begin{equation*}
c_{B,n}^{m}=\log _{2}\left( 1+\frac{p_{B}^{m}\left\vert
h_{B,n}^{m}\right\vert ^{2}}{\Gamma N_{0}}\right)
\end{equation*}%
where $p_{B}^{m}$ denotes the BS's transmission power on the $m^{\text{th}}$
subcarrier and $h_{B,n}^{m}$ is the channel gain from the BS to user $n$
over subcarrier $m$. $\Gamma $ is the gap to Shannon capacity, which is
mainly determined by modulation techniques and the target bit-error rate.
When the BS is communicating with user $n$ assisted by RS $i$ over the $m^{%
\text{th}}$ subcarrier, the achievable rate is given by \cite{Laneman},
\begin{eqnarray}
c_{i,n}^{m} &=&\frac{1}{2}\min \{\log _{2}(1+\frac{p_{B}^{m}\left\vert
h_{B,i}^{m}\right\vert ^{2}}{\Gamma N_{0}}),  \label{9} \\
&&\log _{2}(1+\frac{p_{B}^{m}\left\vert h_{B,n}^{m}\right\vert
^{2}+p_{i,n}^{m}\left\vert h_{i,n}^{m}\right\vert ^{2}}{\Gamma N_{0}})\}
\notag
\end{eqnarray}%
where "$\frac{1}{2}$" in $\left( \ref{9}\right) $ is due to the half duplex
forward constraint, $h_{B,i}^{m}$ is the channel gain from the BS to RS $i$
on subcarrier $m,$ $h_{i,n}^{m}$ is the channel gain from RS $i$ to user $n$
on the $m^{\text{th}}$ subcarrier and $p_{i,n}^{m}$ is the allocated relay
power by RS $i$ for user $n$ over the $m^{\text{th}}$ subcarrier. \ For
notional simplicity we define $H_{B,n}^{m}\hat{=}\left\vert
h_{B,n}^{m}\right\vert ^{2}/\Gamma N_{0},$ $H_{B,i}^{m}\hat{=}\left\vert
h_{B,i}^{m}\right\vert ^{2}/\Gamma N_{0},$ $H_{i,n}^{m}\hat{=}\left\vert
h_{i,n}^{m}\right\vert ^{2}/\Gamma N_{0},(\forall n\in \mathcal{N},i\in
\mathcal{R},m\in \mathcal{M})$ as the normalized channel gains for the rest
of the paper. We use $\mathcal{I}\left( t\right) =\{H_{B,n}^{m}\left(
t\right) ,$ $H_{B,i}^{m}\left( t\right) ,$ $H_{i,n}^{m}\left( t\right)
:\forall n\in \mathcal{N},i\in \mathcal{R},m\in \mathcal{M}\}$ to denote the
collection of\ all the channel state information at time slot $t$. It is
assumed that $\mathcal{I}\left( t\right) $ belongs to a finite\ space that
has arbitrarily large size and is independent and identically distributed
(i.i.d.) at every slot. Considering the BS may communicate with user $n$
directly or with the assistance of selected relay over some subcarriers, the
transmission rate for user $n$ is
\begin{equation}
\mu _{n}\left( t\right) =\sum_{m\in \mathcal{M}}\sum_{i\in \mathcal{R}%
}\left( \alpha _{i,n}^{m}c_{i,n}^{m}+\alpha _{B,n}^{m}c_{B,n}^{m}\right)
\label{44}
\end{equation}

Denote $p_{B}^{\max }$ and $p_{i}^{\max }$ as the total transmission power
constraint for the BS and RSs, respectively, $\mathit{i.e.,}$%
\begin{equation}
p_{B}\leq p_{B}^{\max },\text{ }p_{i}\leq p_{i}^{\max },\;\forall i\in
\mathcal{R},  \label{7}
\end{equation}%
where $p_{B}=\sum_{m\in \mathcal{M}}p_{B}^{m}$, $p_{i}=\sum_{n\in \mathcal{N}%
}\sum_{m\in \mathcal{M}}p_{i,n}^{m}.$ In addition, there are the power mask
constraints$\ \hat{P}\left( m\right) \footnote{%
In reality, the power mask should be location-dependent and each RS and BS
should have different\ mask even over the same subcarrier. Here for easy
expression, we use the same power mask for all RSs and BS over the same
subcarrier. To avoid trivial solutions, it is assumed that $\sum_{m\in
\mathcal{M}}\hat{P}\left( m\right) >p_{B}^{\max },$ $N\sum_{m\in \mathcal{M}}%
\bar{P}\left( m\right) >p_{i}^{\max },\forall i\in \mathcal{R}.$}$ on each
subcarrier for the BS and RSs\ to avoid the adjacent cell interference, $%
\mathit{i.e.,}$%
\begin{equation}
0\leq p_{B}^{m}\leq \hat{P}\left( m\right) ,\;0\leq p_{i,n}^{m}\leq \hat{P}%
\left( m\right) .  \label{45}
\end{equation}%
The determination of $\hat{P}\left( m\right) $ is out of the scope of this
paper and is given a priori.

\subsection{Queueing Model}

At the BS there are random arrival packets with rate $A_{n}\left( t\right) $
packets/slot at the end of\ time slot $t$ waiting for transmission to user $%
n.$ Each user's packets are stored in one of $N$ data queues corresponding
to each destination before they can be sent out. The packet arrival process $%
A_{n}\left( t\right) \in \left[ 0,A_{n}^{\max }\right] $ is assumed to be\
i.i.d over each time slot and $\mathbb{E}\left[ A_{n}\left( t\right) \right]
=\lambda _{n},$ $\forall n\in \mathcal{N}.$\ Among the arrival packets, only
$R_{n}\left( t\right) $ of $A_{n}\left( t\right) $ are admitted into each
user's\ buffer with queue length $Q_{n}\left( t\right) $ in time slot $t$ by
a flow control mechanism. The data queue is updated as
\begin{equation}
Q_{n}\left( t+1\right) =\left[ Q_{n}\left( t\right) -\mu _{n}\left( t\right) %
\right] ^{+}+R_{n}\left( t\right) ,\text{ }\forall n\in \mathcal{N},
\label{33}
\end{equation}%
where $\left[ x\right] ^{+}=\max \left\{ 0,x\right\} $, and $\mu _{n}\left(
t\right) $ is the service rate of user $n$ given in $\left( \ref{44}\right) $%
. To quantitatively measure the balance between the transmission request $%
R_{n}\left( t\right) $ and service rate $\mu _{n}\left( t\right) ,$ the
definition of stability is given first.

\begin{definition}[\protect\cite{Neely_book2006}]
A queue $Q\left( t\right) $ is strongly stable if:
\begin{equation*}
\underset{t\rightarrow \infty }{\lim \sup }\frac{1}{t}\sum_{\tau =0}^{t-1}%
\mathbb{E}\left\{ Q\left( \tau \right) \right\} <\infty
\end{equation*}%
A network is strongly stable if all individual queues of the network are
strongly stable.
\end{definition}

In the following, strong stability is also referred to stability.

\subsection{Energy Management Model}

Actually, the power consumption of BS or\ RSs consists of two parts. One
part is the dynamic\ transmission power of the power amplifier denoted as $p$
and the other part is $\Delta p,$\ the static power consumption of the radio
frequency chains including power dissipation in the mixer, transmit filters
and digital-to-analog converters \cite{ng} \cite{HuangHUang}. Thus the power
consumption of the BS and relay $i$ can be expressed as $P_{B}=p_{B}+\Delta
p_{B}$ and $P_{i}=p_{i}+\Delta p_{i}$, respectively. Since the BS dominates
the energy consumption in cellular networks \cite{JianTang_infocom2014}, its
required energy can be taken not only from the renewable energy generator
but also from the power grid. Due to the investment cost each RS is only
powered by grid. To "smooth" the intermittent renewable energy supply, we
introduce an energy storage device to the BS. The harvested renewable energy
is first put into the storage device and then withdrawn to supply the BS. If
the residual energy in the storage is not enough, the BS will be supplied by
the grid. Since the storage device has a limited size, we have to cope with
charging and discharging properly. Let $w\left( t\right) $ with $w\left(
t\right) \in \left[ 0,w^{\max }\right] $ denote the harvested energy from
the\ renewable source at time slot $t$. Only $\delta \left( t\right) w\left(
t\right) $ of the renewable energy is actually charged into the storage
device, where $\delta \left( t\right) \in \left[ 0,1\right] $.\ The energy
level $S\left( t\right) $ at the storage device is updated as an energy
queue \cite{MIng_li}:
\begin{equation}
S\left( t+1\right) =S\left( t\right) -O\left( t\right) +\delta \left(
t\right) w\left( t\right) ,  \label{3}
\end{equation}%
where $O\left( t\right) $ is the total power withdrawn from the storage
device to supply the BS at time slot $t$. As it can be found later, the
energy management of the BS is designed based on $S\left( t\right) $ instead
of the profile of $w\left( t\right) .$\ There is a maximum discharge rate
constraint $O^{\max }$ for the storage device.\ The energy level in the
storage device should be always available and can not exceed the storage
capacity $S_{{}}^{\max }$. Therefore, at each time slot, we should ensure
that
\begin{equation}
O\left( t\right) \leq \min \left\{ S\left( t\right) ,O^{\max }\right\} .
\label{4}
\end{equation}%
\begin{equation}
w\left( t\right) \leq \min \left\{ w^{\max },S_{{}}^{\max }-S\left( t\right)
\right\} .  \label{5}
\end{equation}

The BS's power consumption $P_{B}\left( t\right) $ is supplied by $O\left(
t\right) $ from the storage device and $J\left( t\right) $ from power grid
directly with $0\leq J\left( t\right) \leq J^{\max }$, \textit{i.e., }%
\begin{equation}
P_{B}\left( t\right) =O\left( t\right) +J\left( t\right) .  \label{6}
\end{equation}%
$J^{\max }$ is used to take into account the hardware and line constraint of
BS power supply\ system. Thus the total energy withdrawn from the power grid
at time slot $t$ is $P\left( t\right) =J\left( t\right) +\sum_{i\in \mathcal{%
R}}P_{i}\left( t\right) $. For the BS being able to be supplied by storage
device or power grid independently, it is assumed that $J_{{}}^{\max }\geq
p_{B}^{\max }+\Delta p_{B}$ and $O_{{}}^{\max }\geq p_{B}^{\max }+\Delta
p_{B}.$ The important notations of this paper are listed in Table \ref{table
1}.

\begin{center}
\begin{table}[tbp] \centering%
\caption{Notations}\label{table 1}
\begin{tabular}{|l|l|}
\hline
Symbol & Definition \\ \hline
$\mathcal{N}=\left\{ 1,\cdots ,N\right\} $ & Set of users \\
$\mathcal{R}=\left\{ 1,\cdots ,K\right\} $ & Set of RSs \\
$\mathcal{M}=\left\{ 1,\cdots ,M\right\} $ & Set of subcarriers \\
$\alpha _{B,n}^{m},\alpha _{i,n}^{m}$ & Subcarrier allocation indicator \\
$c_{B,n}^{m},c_{i,n}^{m}$ & Direct and cooperative transmission rates \\
$H_{B,n}^{m},H_{B,i}^{m},H_{i,n}^{m}$ & Normalized channel gains \\
$\mathcal{I}\left( t\right) $ & Collection of all channel states \\
$\mu _{n}\left( t\right) $ & Transmission rate for user $n$ \\
$p_{B}^{m}$ & Transmission power of BS over subcarrier $m$ \\
$p_{i,n}^{m}$ & Transmission power of RS over subcarrier $m$ \\
$X_{n}\left( t\right) ,R_{n}\left( t\right) $ & Potential admitted rate and
actual admitted rate \\
$Q_{n}\left( t\right) $ & Queue length of user $n$ at BS \\
$U_{n}\left( t\right) $ & Virtual queue length of user $n$ at BS \\
$S\left( t\right) $ & Energy level at BS storage device \\
$O\left( t\right) $ & Withdraw power from storage device \\
$\omega \left( t\right) ,\delta \left( t\right) \omega \left( t\right) $ &
Harvested and actual stored renewable energy \\
$J\left( t\right) $ & Energy drawn from the grid to supply BS \\ \hline
\end{tabular}%
\end{table}%
\vspace{-0.5cm} 
\end{center}

\section{Problem Formulation}

The objective of the network operator is to maximize the total utility of
the average throughput for each user with the least on-grid energy
consumption. Here the utility function $f\left( \cdot \right) $ is assumed
to be non-decreasing, non-negative and concave, which is used to maintain
fairness in resource allocation \cite{Neely_book2006}. Since there are some
random events in the system, such as random arrival packets towards each
user and renewable energy generation,\ it is more reasonable to optimize the
objective in a long-term manner. Basically the system stability should be
maintained. Since the maximum system utility and the least on-grid energy
consumption are usually two conflicting objectives, we will find a policy to
solve the following stochastic optimization problem to achieve optimal
tradeoff between them,

\begin{equation*}
\left( P1\right)
\begin{array}{ll}
\max & \phi \sum_{n\in \mathcal{N}}f\left( \bar{R}_{n}\right) -\varphi \bar{P%
} \\
\text{s.t.} & \text{C1: }\bar{R}_{n}\leq \lambda _{n},\text{ }\forall n\in
\mathcal{N} \\
& \text{C2: }\left( \ref{2}\right) ,\left( \ref{7}\right) ,\left( \ref{45}%
\right) ,\left( \ref{3}\right) ,\left( \ref{4}\right) ,\left( \ref{5}\right)
,\left( \ref{6}\right) ,\forall t \\
& \text{C3: the system is stable}%
\end{array}%
\end{equation*}%
where
\begin{equation*}
\bar{R}_{n}\hat{=}\lim_{t\rightarrow \infty }\frac{1}{t}\sum_{\tau =0}^{t-1}%
\mathbb{E}\left\{ R_{n}\left( \tau \right) \right\}
\end{equation*}%
is the time-average admitted rate of user $n$, and
\begin{equation*}
\bar{P}\hat{=}\lim_{t\rightarrow \infty }\frac{1}{t}\sum_{\tau =0}^{t-1}%
\mathbb{E}\left\{ P\left( \tau \right) \right\}
\end{equation*}%
is the time-average on-grid energy consumption and the expectations are
taken with respect to the random events. $\phi $ and $\varphi $ are two
constants introduced to make a tradeoff between throughput utility and
on-grid energy consumption. Denote the optimal objective value of $\left(
P1\right) $ as $\mathcal{V}_{1}.$ Due to C1 and the bounded on-grid energy
consumption, we can always find a constant $\mathcal{V}^{\max }$ such that $%
\mathcal{V}_{1}\leq \mathcal{V}^{\max }.$ In $\left( P1\right) $ if $f\left(
\bar{R}_{n}\right) $ is chosen as $\bar{R}_{n},$\ the objective function of $%
\left( P1\right) $ is also known as the parametric form in solving the
energy efficiency maximization problem \cite{Sheng_Min}.

In problem $\left( P1\right) ,$ the current energy management policy is
coupled with future ones due to the iteration $\left( \ref{3}\right) ,$
\emph{i.e.}, the current energy management policy may have impacts on the
future energy charging or discharging action for the energy storage device.
This coupling nature together with unknown statistics of renewable energy
generation makes $\left( P1\right) $ difficult. Denote $\bar{O}\hat{=}%
\underset{t\rightarrow \infty }{\lim }\frac{1}{t}\sum\limits_{\tau =0}^{t-1}%
\mathbb{E}\left\{ O\left( \tau \right) \right\} ,$ $\bar{\delta}\bar{w}\hat{=%
}\underset{t\rightarrow \infty }{\lim }\frac{1}{t}\sum\limits_{\tau =0}^{t-1}%
\mathbb{E}\{\delta \left( \tau \right) w\left( \tau \right) \}$ as the
time-average value of the expected energy withdrawn from the storage device
and the input\ renewable energy, respectively. By summing $\left( \ref{3}%
\right) $ from the\ initial state to $t-1$ and taking expectations on both
sides,\ we have
\begin{equation}
\mathbb{E}\left\{ S\left( t\right) \right\} -S\left( 0\right) =\sum_{\tau
=0}^{t-1}\mathbb{E}\left\{ -O\left( \tau \right) +\delta \left( \tau \right)
w\left( \tau \right) \right\} .  \label{8}
\end{equation}%
Since the energy level in the storage device satisfies $0\leq S_{{}}\left(
t\right) \leq S_{{}}^{\max }$, we have $\bar{O}=\bar{\delta}\bar{w}$ after
dividing both sides of $\left( \ref{8}\right) $ by $t$\ and\ taking
limitation of $t\rightarrow \infty $. Since the objective function of $%
\left( P1\right) $ is a function of the time average, to make it tractable
we transfer it to an optimization problem with a time averaged objective\
function. We instead consider $\left( P2\right) $ to tackle the
time-coupling difficulties in $\left( \ref{3}\right) $ and the\ objective
function in $\left( P1\right) $.%
\begin{equation*}
\left( P2\right)
\begin{array}{ll}
\max & \phi \sum_{n\in \mathcal{N}}\overline{f\left( X_{n}\right) }-\varphi
\bar{P} \\
\text{s.t.} & \text{C1': }\bar{X}_{n}\leq \bar{R}_{n},\text{ }\forall n\in
\mathcal{N} \\
& \text{C2': }X_{n}\left( t\right) \leq A_{n}^{\max },\forall t,n\in
\mathcal{N} \\
& \text{C3': }\left( \ref{2}\right) ,\left( \ref{7}\right) ,\left( \ref{45}%
\right) ,\left( \ref{5}\right) ,\left( \ref{6}\right) ,\forall t \\
& \text{C3: the system is stable} \\
& \text{C4: }\bar{O}=\bar{\delta}\bar{w}%
\end{array}%
\end{equation*}%
where $X_{n}\left( t\right) $ is an auxiliary variable, and $\overline{%
\!f\left( \!X_{n}\!\right) }\!\hat{=}\!\lim_{t\!\rightarrow \!\infty }\!%
\frac{1}{t}\!\sum_{\!\tau =0}^{t-1}\!\mathbb{E}\!\!\left\{ \!f\!\left(
\!X_{n}\!\left( t\!\right) \!\right) \!\right\} ,$ $\!\bar{X}_{n}\!$ $\!\hat{%
=}\!\lim_{\!t\rightarrow \!\infty }\!\!\frac{1}{t}\!\!\sum_{\!\tau
=0\!}^{t-1}\!\mathbb{E}\!\left\{ \!X_{n}\!\left( \!t\right) \!\right\} ,$
which can be understood as the lower bound of $\bar{R}_{n}$ as shown in C1'.
In the following section, we will use virtual queue technique to satisfy the
time average constraint C1'.\ Denote the optimal objective value of $\left(
P2\right) $ as $\mathcal{V}_{2}$. By Ch5 of \cite{Neely_book_2010}, we have $%
\mathcal{V}_{2}\geq \mathcal{V}_{1}.$ Now, there is no time correlation on
the energy level of storage device in $\left( P2\right) $. Using the theory
of Lyapunov optimization, $\left( P2\right) $ can be solved by an online
algorithm to achieve the close-to-optimal performance subject to system
stability.

\section{Joint Flow Control, Resource Allocation and Energy Management}

Since the buffer size at the BS is finite, to ensure that each data queue $%
Q_{n}\left( t\right) $ at the BS has a deterministic upper bound $Q^{\max }$
with $Q^{\max }\geq A_{{}}^{\max },$ where $A_{{}}^{\max }=\max_{n}\left\{
A_{n}^{\max }\right\} $, we introduce the following virtual queue,
\begin{equation}
U_{n}\left( t+1\right) =\left[ U_{n}\left( t\right) -R_{n}\left( t\right) %
\right] ^{+}+X_{n}\left( t\right) ,\text{ }\forall n\in \mathcal{N}
\label{1}
\end{equation}%
whose stability can satisfy the constraint C1' in $\left( P2\right) .$

It is observed that in $\left( P2\right) ,$ the decision variables $\left\{
R_{n}\left( t\right) \right\} ,\left\{ \alpha _{B,n}^{m},\alpha
_{i,n}^{m}\left( t\right) \right\} ,\left\{ P\left( t\right) \right\}
,\left\{ O\left( t\right) \right\} $ and $\left\{ \delta \left( t\right)
\right\} $ are coupled together. We decompose the system problem $\left(
P2\right) $ into joint\ Flow control,\ Resource allocation and Energy
managEment (\emph{FREE}) by Lyapunov optimization method in each time slot.
Its design principle is to minimize the upper bound of Lyapunov drift minus
system objective function greedily without knowing the profile of\
stochastic events. The detailed derivation can be found in the proof of
Theorem \ref{Theo:Nearoptimal}. Fig. \ref{Fig. 11} illustrates that the
coupled dynamics in \emph{FREE}\ are coordinated by queue states $\left\{
U_{n}\left( t\right) ,X_{n}\left( t\right) ,S\left( t\right) \right\} .$
Although \emph{FREE }are designed to solve $\left( P2\right) $ online and
separately, Theorem \ref{Theo:Nearoptimal} demonstrates that it can approach
the optimal solution of $\left( P1\right) $ asymptotically. Its detailed
design is as follows. In the following, the joint flow control, wireless
resource allocation and energy management in \emph{FREE}\ will be described
in detail. %
\begin{figure}[tbph]
\centering
\includegraphics[width=9.5cm]{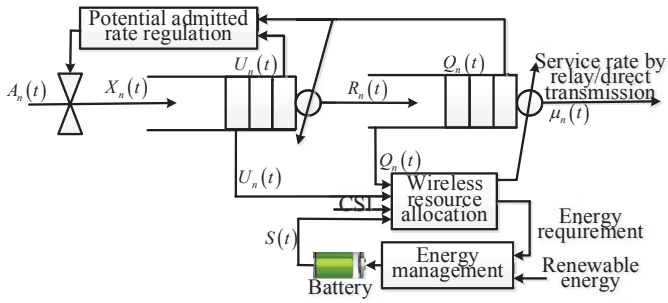}
\caption{Diagram of coupled problems of flow control, wireless resource
allocation and energy management.}
\label{Fig. 11}
\end{figure}

\subsection{Flow Control}

At each time slot $t$, the BS decides the admission data rate $\left\{
R_{n}\left( t\right) \right\} $ to\ the data buffer of each user as the
solution to the following problem,
\begin{equation}
\begin{array}{ll}
\min_{R_{n}\left( t\right) } & R_{n}\left( t\right) \left( Q_{n}\left(
t\right) -Q^{\max }+A^{\max }\right) \\
\text{s.t.} & 0\leq R_{n}\left( t\right) \leq A_{n}\left( t\right)%
\end{array}%
.  \label{59}
\end{equation}%
The problem $\left( \ref{59}\right) $ is\ corresponding to one part of the
upper bound of Lyapunov drift plus penalty in $\left( \ref{43}\right) .$\
The solution to the above problem is
\begin{equation}
R_{n}\left( t\right) =\left\{
\begin{array}{cc}
0 & \text{if }Q_{n}\left( t\right) >Q^{\max }-A^{\max } \\
A_{n}\left( t\right) & \text{otherwise}%
\end{array}%
\right. .  \label{32}
\end{equation}%
In $\left( \ref{32}\right) ,$ the tradeoff between system utility and
on-grid energy is not explicitly considered. It can be intuitively
understood that more saved on-grid energy will result in a smaller $\mu
_{n}\left( t\right) $ and a larger $Q_{n}\left( t\right) $ subsequently by $%
\left( \ref{33}\right) $. Then according to $\left( \ref{32}\right) ,$ more
arrival packets will be dropped in order to mitigate congestion. Thus,
on-grid energy is saved with a compromise of throughput.

The potential admitted rate $X_{n}\left( t\right) $ is the solution to the
following problem,
\begin{equation}
\begin{array}{ll}
\min_{X_{n}\left( t\right) } & \frac{Q^{\max }-A^{\max }}{Q^{\max }}%
U_{n}\left( t\right) \cdot X_{n}\left( t\right) \!-\!V\phi f\left(
X_{n}\left( t\right) \right) \\
\text{s.t.} & 0\leq X_{n}\left( t\right) \leq A^{\max }%
\end{array}%
.  \label{60}
\end{equation}%
If $f\left( \cdot \right) $ is differentiable, $X_{n}\left( t\right) $ is
updated as
\begin{equation}
X_{n}\!\left( t\!\right) \!=\!\min \!\left\{ \!\max \!\left\{ \!0,f^{\prime
-1}\!\left( \!\frac{Q^{\max }\!-\!A^{\max }}{Q^{\max }V\phi }U_{n}\!\left(
t\right) \!\right) \!\right\} ,\!A^{\max }\!\right\} ,  \label{34}
\end{equation}%
where $V\geq 0$ is a parameter set according to the system performance
requirement and is to be specified in Section VI.

\subsection{Wireless Resource Allocation}

We determine the power control, RS selection and subcarrier allocation by
solving the following optimization problem based on the measurement of
channel states and queue states at the beginning of each time slot.
\begin{equation}
\begin{array}{cc}
\max &
\begin{array}{c}
\sum_{n\in \mathcal{N}}^{{}}\left( \frac{U_{n}\left( t\right) Q_{n}\left(
t\right) }{Q^{\max }}\mu _{n}\left( t\right) +\left( S\left( t\right)
-\theta \right) p_{B}\left( t\right) \right) \\
-V\varphi \sum_{i\in \mathcal{R}}^{{}}p_{i}\left( t\right)%
\end{array}
\\
\text{s.t.} & \left( \ref{2}\right) ,\left( \ref{7}\right) ,\left( \ref{45}%
\right)%
\end{array}%
.  \label{23}
\end{equation}%
The optimization problem $\left( \ref{23}\right) $ has zero-duality-gap if
the number of subcarrier goes to infinity \cite{Weiyu_2007} \cite%
{JianweiHUang}.\ We introduce the Lagrange multipliers $\lambda _{i}$ and $%
\lambda _{B}$ to relax the power constraints $\left( \ref{7}\right) $ and
solve $\left( \ref{23}\right) $ by Lagrangian\ dual decomposition. The
Lagrange function is
\begin{eqnarray}
&&L\left( \mathbf{\alpha },\mathbf{p,\lambda }\right) \\
&=&\sum_{n=1}^{N}(\frac{U_{n}\left( t\right) Q_{n}\left( t\right) }{Q^{\max }%
}\mu _{n}\left( t\right) +\left( S\left( t\right) -\theta \right)
p_{B}\left( t\right) )  \notag \\
&&-V\!\varphi\!\sum_{i=1}^{K}\!p_{i}\!\left( t\right)\!+\!\lambda
_{B}\!\left( p_{B}^{\max }\!-\!p_{B}\!\left(
t\right)\!\right)\!+\!\sum_{i=1}^{K}\!\lambda _{i}\!\left( p_{i}^{\max
}\!-\!p_{i}\!\left(t\right)\!\right) .  \notag
\end{eqnarray}%
The dual function is $D\left( \mathbf{\lambda }\right) =\max_{\mathbf{\alpha
,p}}L\left( \mathbf{\alpha },\mathbf{p,\lambda }\right) \footnote{%
Hereafter, we use bold letter to denote a vector, \textit{e.g., }$\mathbf{%
\lambda =}\left[ \lambda _{1},\lambda _{2},\cdots ,\lambda _{K},\lambda _{B}%
\right] .$},$ which can be decomposed into $M$ independent subproblems,%
\begin{equation}
\begin{array}{cc}
\max & \sum_{n\in \mathcal{N}}\sum_{i\in \mathcal{R}}\Omega _{m} \\
\text{s.t. } & \left( \ref{45}\right)%
\end{array}%
,\forall m\in \mathcal{M},  \label{10}
\end{equation}%
where
\begin{eqnarray}
&&\Omega _{m}  \label{46} \\
&=&\frac{U_{n}Q_{n}}{Q^{\max }}\alpha _{B,n}^{m}c_{B,n}^{m}\left( t\right)
+\alpha _{B,n}^{m}\left( S\left( t\right) -\theta -\lambda _{B}\right)
p_{B}^{m}\left( t\right)  \notag \\
&\hat{=}&\Omega _{m}\left( B,n\right)  \notag
\end{eqnarray}%
if direct transmission towards user $n$ is adopted over subcarrier $m$, or%
\begin{eqnarray}
&&\Omega _{m}  \label{49} \\
&=&\frac{U_{n}Q_{n}}{Q^{\max }}\alpha _{i,n}^{m}c_{i,n}^{m}\left( t\right)
+\alpha _{i,n}^{m}\left( S\left( t\right) -\theta -\lambda _{B}\right)
p_{B}^{m}\left( t\right)  \notag \\
&&-\alpha _{i,n}^{m}\left( V\varphi +\lambda _{i}\right) p_{i,n}^{m}\left(
t\right)  \notag \\
&\hat{=}&\Omega _{m}\left( i,n\right) ,  \notag
\end{eqnarray}%
if relay-assisted transmission is chosen.

We first consider the direct transmission $\left( \ref{46}\right) $ to find
the optimal $\Omega _{m}^{\ast }\left( B,n\right) ,$%
\begin{equation*}
\begin{array}{cc}
\underset{0\leq p_{B}^{m}\leq \bar{P}\left( m\right) }{\max } & \frac{%
U_{n}\left( t\right) Q_{n}\left( t\right) }{Q^{\max }}c_{B,n}^{m}\left(
t\right) +\left( S\left( t\right) -\theta -\lambda _{B}\right)
p_{B}^{m}\left( t\right)%
\end{array}%
.
\end{equation*}%
The optimal\ transmission power of the BS over subcarrier $m$ is obtained as
\begin{equation}
p_{B}^{m}\left( t\right) =\left\{
\begin{array}{ll}
\hat{p}_{B}^{m}\left( t\right) & \text{otherwise} \\
\bar{P}\left( m\right) & \text{if }\theta +\lambda _{B}-S\left( t\right)
\leq 0%
\end{array}%
\right. ,  \label{47}
\end{equation}%
where $\hat{p}_{B}^{m}\left( t\right) =\left( \frac{U_{n}\left( t\right)
Q_{n}\left( t\right) }{Q^{\max }\left( \theta +\lambda _{B}-S\left( t\right)
\right) \log 2}-\frac{1}{H_{B,n}^{m}}\right) _{0}^{\bar{P}\left( m\right) }$
and$\ \left( x\right) _{a}^{b}=\min \left\{ b,\max \left\{ a,x\right\}
\right\} .$ In the case of $\theta +\lambda _{B}-S\left( t\right) >0$, $%
p_{B}^{m}\left( t\right) $ in $\left( \ref{47}\right) $ is obtained by
solving the derivative of $\Omega _{m}\left( B,n\right) $ to $%
p_{B}^{m}\left( t\right) $ being equal to zero. After substituting $\left( %
\ref{47}\right) $ into $\left( \ref{46}\right) $ we can obtain the optimal $%
\Omega _{m}^{\ast }\left( B,n\right) $ given that the BS transmits to user $%
n $ over subcarrier $m$ directly.

In the case of cooperative transmission, we consider the following two cases
$\left( \ref{11}\right) $ and $\left( \ref{12}\right) $ for the objective
function $\left( \ref{49}\right) $.
\begin{equation}
\begin{array}{cc}
\max &
\begin{array}{c}
q_{n}\left( t\right) \log \left(
1+p_{B}^{m}H_{B,n}^{m}+p_{i,n}^{m}H_{i,n}^{m}\right) \\
+\left( S\left( t\right) -\theta -\lambda _{B}\right) p_{B}^{m}-\left(
V\varphi +\lambda _{i}\right) p_{i,n}^{m}%
\end{array}
\\
\text{s.t.} &
p_{B}^{m}H_{B,i}^{m}>p_{B}^{m}H_{B,n}^{m}+p_{i,n}^{m}H_{i,n}^{m} \\
& 0\leq p_{B}^{m}\left( t\right) \leq \bar{P}\left( m\right) ,\;0\leq
p_{i,n}^{m}\left( t\right) \leq \bar{P}\left( m\right)%
\end{array}%
,  \label{11}
\end{equation}%
\begin{equation}
\begin{array}{cc}
\max &
\begin{array}{c}
q_{n}\left( t\right) \log \left( 1+p_{B}^{m}H_{B,i}^{m}\right) \\
+\left( S\left( t\right) -\theta -\lambda _{B}\right) p_{B}^{m}-\left(
V\varphi +\lambda _{i}\right) p_{i,n}^{m}%
\end{array}
\\
\text{s.t.} & p_{B}^{m}H_{B,i}^{m}\leq
p_{B}^{m}H_{B,n}^{m}+p_{i,n}^{m}H_{i,n}^{m} \\
& 0\leq p_{B}^{m}\left( t\right) \leq \bar{P}\left( m\right) ,\;0\leq
p_{i,n}^{m}\left( t\right) \leq \bar{P}\left( m\right)%
\end{array}%
,  \label{12}
\end{equation}%
where $q_{n}\left( t\right) =\frac{U_{n}\left( t\right) Q_{n}\left( t\right)
}{2Q^{\max }\log 2}.$

\begin{theorem}
\label{Theo:opt_power} The optimal transmission power allocation for problem
$\left( \ref{11}\right) $ is as follows:
\end{theorem}

\begin{itemize}
\item If $S\left( t\right) -\theta -\lambda _{B}+\frac{\left( V\varphi
+\lambda _{i}\right) H_{B,n}^{m}}{H_{i,n}^{m}}\geq 0,$ $p_{B}^{m}\left(
t\right) =\bar{P}\left( m\right) ,\ p_{i,n}^{m}\left( t\right) =\left( \frac{%
d_{i,n}^{m}-\bar{P}\left( m\right) H_{B,n}^{m}}{H_{i,n}^{m}}\right) ^{+},\ $%
where $d_{i,n}^{m}=\max \{0,\min \{\frac{q_{n}\left( t\right) H_{i,n}^{m}}{%
V\varphi +\lambda _{i}}-1,\bar{P}\left( m\right) H_{B,i}^{m},\bar{P}\left(
m\right) H_{i,n}^{m}+\bar{P}\left( m\right) H_{B,n}^{m}\}\}.$

\item If $S\left( t\right) -\theta -\lambda _{B}+\frac{\left( V\varphi
+\lambda _{i}\right) H_{B,n}^{m}}{H_{i,n}^{m}}<0,$ $p_{B}^{m}\left( t\right)
=p_{i,n}^{m}\left( t\right) =0.$
\end{itemize}

The optimal solution to problem $\left( \ref{12}\right) $ is

\begin{itemize}
\item If $H_{B,i}^{m}>H_{B,n}^{m},$ $p_{B}^{m}\left( t\right)
=p_{i,n}^{m}\left( t\right) =0.$

\item If $H_{B,i}^{m}\leq H_{B,n}^{m},$ $p_{i,n}^{m}\left( t\right) =0,$ $%
p_{B}^{m}=\left\{
\begin{array}{cc}
\bar{P}\left( m\right) & \text{if }S\left( t\right) -\theta -\lambda
_{B}\geq 0 \\
\left( \frac{q_{n}\left( t\right) }{\lambda _{B}-S\left( t\right) +\theta }-%
\frac{1}{H_{B,i}^{m}}\right) _{0}^{\bar{P}\left( m\right) } & \text{otherwise%
}%
\end{array}%
\right. $.
\end{itemize}

\begin{proof}
The proof can be found in Appendix A.
\end{proof}

Some intuitive explanations for \emph{Theorem \ref{Theo:opt_power}} are
given below.


\begin{remark}
\textit{\ }For the solution to problem $\left( \ref{11}\right) $, it can be
found that in the case of $S\left( t\right) -\theta -\lambda _{B}+\frac{%
\left( \varphi +\lambda _{i}\right) H_{B,n}^{m}}{H_{i,n}^{m}}\geq 0$ if the
second hop channel gain $H_{i,n}^{m}$ is relatively low compared with $%
H_{B,n}^{m}$, which may result in $d_{i,n}^{m}=0,$ the BS will ask no relay
for help transmission to avoid more electricity withdrawn from power\ grid
to supply the relay. Moreover, for the same case, if $q_{n}\left( t\right) $
is large enough, which may result in a large $d_{i,n}^{m},$ the BS has to
ask relay $i$ for helping transmission to alleviate congestion of user $n$'s
buffer. In the aforementioned situation, the BS transmits with $\bar{P}%
\left( m\right) $ since there is enough residual energy in the storage
device and the objective function of $\left( \ref{11}\right) $ is increasing
with respect to $p_{B}^{m}.$ If $S\left( t\right) -\theta -\lambda _{B}+%
\frac{\left( \varphi +\lambda _{i}\right) H_{B,n}^{m}}{H_{i,n}^{m}}<0$ due
to less residual energy in the storage device, the combination of\ relay $i,$
the BS and user $n$ is not suitable for transmission over subcarrier $m$ to
avoid withdrawing too much electricity from the grid.
\end{remark}


\begin{remark}
\textit{\ }By the above description, it may be concluded some subcarriers
will be abandoned temporarily due to their low energy efficiency. At the
same time the abandoned subcarriers can be used by other cells belonging to
the same operator. The abandoned subcarriers can also be re-used by the same
cell after a while when the channel conditions turn to be better.\
Therefore, the network operator will not waste any resource or revenue from
a long term and global point of view.
\end{remark}

By substituting the optimal transmission power in \emph{Theorem \ref%
{Theo:opt_power}} to $\left( \ref{49}\right) $ the optimal $\Omega
_{m}^{\ast }\left( i,n\right) $ is obtained. Then the subcarrier allocation
and relay selection can be achieved by
\begin{eqnarray}
&&\alpha _{j,n}^{m}  \label{58} \\
&\!=\!&\!\left\{
\begin{array}{ll}
1 & \text{if }\!\left( \!j,\!n\!\right) \!\!=\!\!\underset{_{j\!\in
\!\left\{ \!B\!\right\} \!\cup \!\mathcal{R},\!n\in \mathcal{N}}}{\arg }\max
\!\left\{ \!\Omega _{m}^{\ast }\!\left( i,n\!\right) ,\!\Omega _{m}^{\ast
}\!\left( \!B,n\!\right) \!\right\} \\
0 & \text{otherwise}%
\end{array}%
\right.  \notag
\end{eqnarray}%
Note that in $\left( \ref{58}\right) $ $j$ can be the index of a RS or
referred to the BS if direct transmission is adopted.

The associated dual problem of $\left( \ref{23}\right) $ is $\min_{\mathbf{%
\lambda \preceq 0}}D\left( \mathbf{\lambda }\right) ,$ which can be solved
by standard subgradient method $\left( \ref{21}\right) $ and $\left( \ref{22}%
\right) $ $,\forall i\in \mathcal{R}$. \vspace{-0.2cm}
\begin{eqnarray}
&&\lambda _{B}\left( k+1,t\right)  \label{21} \\
&=&\left[ \lambda _{B}\left( k,t\right) +\varsigma _{k}\left(
\sum_{m=1}^{M}p_{B}^{m}\left( k,t\right) -p_{B}^{\max }\right) \right] ^{+}
\notag
\end{eqnarray}%
\begin{eqnarray}
&&\lambda _{i}\left( k+1,t\right)  \label{22} \\
&=&\left[ \lambda _{i}\left( k,t\right) +\varsigma _{k}\left(
\sum_{n=1}^{N}\sum_{m=1}^{M}\alpha _{i,n}^{m}p_{i,n}^{m}\left( k,t\right)
-p_{i}^{\max }\right) \right] ^{+},\text{ }  \notag
\end{eqnarray}%
where $k$ is the iteration index in time slot $t,$ $\left[ x\right] ^{+}\hat{%
=}\max \left\{ 0,x\right\} $ and $\varsigma _{k}$ is step size at the $k^{%
\text{th}}$ iteration. The convergence of $\left( \ref{21}\right) $ and $%
\left( \ref{22}\right) $ can be guaranteed by a diminishing step size $%
\varsigma _{k}$ \cite{Mikael_johansson}.


\begin{remark}
In practice, the subproblems $\left( \ref{11}\right) $ and $\left( \ref{12}%
\right) $ are solved by each RS locally for $NM$ times during one iteration $%
k$ and the computational complexity at each relay is $O\left( NM\right) ,$
which is the same as \cite{Weiyu_2007}. And then each relay transmits the
optimal $\Omega _{m}^{\ast }\left( i,n\right) $ to the BS for the final
subcarrier allocation.
\end{remark}

The deployment of a certain amount RSs to proper place is also important in
utilizing the cooperative diversity and has been extensively investigated in
\cite{Shen_relay_renew} and \cite{NiuZhisheng}. A candidate place is where a
RS can not only serve users better than direct transmission but also be easy
to access on-grid power. Candidate place can be the roof of a building,
where RSs serve users in the building. Thus the number of candidate places
in a considered area is finite. Since the joint deployment of RS and
resource allocation problem is NP-hard, heuristic algorithm will be
considered. Firstly, all candidate places are deployed with RSs. Users
requiring relay are associated to the RS with the minimum \emph{RTE} which
is defined as the quotient of user's throughput requirement and achievable
throughput with unit transmission power and subcarrier. Then subcarriers are
allocated among RSs and BS to satisfy each user's throughput requirement
proportionally. The transmission power of BS and RSs can be obtained by
waterfilling. Secondly, calculate each RS's energy efficiency value, which
is defined as the total achievable throughput divided by the RS's energy
consumption. Sort all RSs in an ascending order based on its energy
efficiency value. Thirdly, try to delete the first RS from and candidate
place and re-connect its users to other RSs based on the \emph{RTE}
criterion in the first step. Subcarriers and transmission power are
re-allocated. If each user's throughput requirement is satisfied, try to
delete the second RS and repeat. If the removal of the $\left( k+1\right) ^{%
\text{th}}$ RS cannot satisfy users throughput requirements, the algorithm
outputs the deployment of $k$ RSs in the last round. Due to space
limitation, we do not evaluate the heuristic RS placement. Its idea is
similar with \cite{Shen_relay_renew}, which has demonstrated good
performance with low time complexity.

\subsection{Energy Management}

At each time slot $t$, the BS manages the energy level in its storage device
based on local information by solving the following problem:%
\begin{equation}
\begin{array}{cc}
\min_{J,\delta } & \left( \varphi V+S\left( t\right) -\theta \right) J\left(
t\right) +\left( S\left( t\right) -\theta \right) \delta \left( t\right)
w\left( t\right) \\
\text{s.t. } & \left( \ref{5}\right) ,\text{ }\left( \ref{6}\right) ,\text{ }%
\delta \left( t\right) \in \left[ 0,1\right]%
\end{array}%
,  \label{24}
\end{equation}%
where $\theta $ is a parameter to be specified in Section VI. The optimal
solution to $\left( \ref{24}\right) $ has the following threshold structure:

1) If $S\left( t\right) \geq \theta -\varphi V,$ then $J_{{}}^{\ast }=0,$ $%
O_{{}}^{\ast }=\min \left\{ p_{B}^{\ast }+\Delta p_{B}^{{}},O_{{}}^{\max
}\right\} ,$ where $p_{B}^{\ast }$ is the solution to $\left( \ref{23}%
\right) \ $and
\begin{equation}
\delta _{{}}^{\ast }\left( t\right) =\left\{
\begin{array}{cc}
1 & \text{if }0\leq S\left( t\right) <\theta \\
0 & \text{otherwise}%
\end{array}%
\right..  \label{61}
\end{equation}

2) If $S\left( t\right) <\theta -\varphi V,$ then $J_{{}}^{\ast }=\min
\left\{ p_{B}^{\ast }+\Delta p_{B}^{{}},J_{{}}^{\max }\right\} ,$ $%
O_{{}}^{\ast }=\max \left\{ 0,p_{B}^{\ast }+\Delta p_{B}^{{}}-J^{\ast
}\right\} ,$ $\delta _{{}}^{\ast }\left( t\right) =1.$

The solution to $\left( \ref{24}\right) $ depends on the current energy
level in the energy storage device without knowing any statistic information
on renewable energy.

\begin{remark}
The detailed implementation of the proposed joint policy \emph{FREE} can be
found in Fig. \ref{Fig. 1}. Although the \emph{FREE} policy seems to be
implemented separately, The proof of Theorem \ref{Theo:Nearoptimal}
demonstrates that the solutions to subproblems $\left( \ref{59}\right) ,$ $%
\left( \ref{60}\right) ,$ $\left( \ref{23}\right) $ and $\left( \ref{24}%
\right) $ minimize the upper bound of system Lyapunov drift plus cost
function greedily. Thus, the\textit{\ FREE }can solve the system
optimization problem $\left( P1\right) $ while maintaining stability at the
same time, which is formally stated in Proposition \ref{Prop:boundsQ} and
Theorem \ref{Theo:Nearoptimal}. 
\begin{figure}[tbph]
\centering
\includegraphics[width=8.5cm]{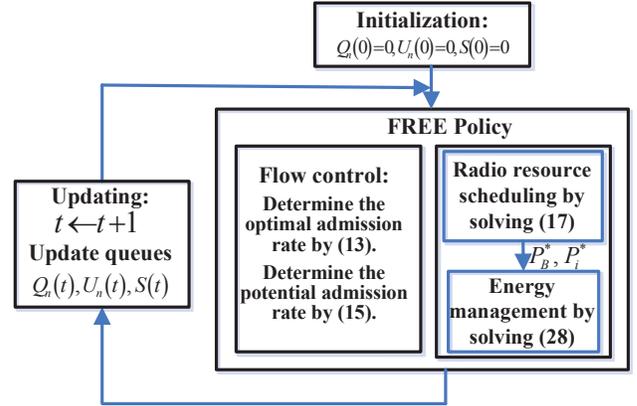}
\caption{The system procedure of the \textit{FREE} scheme.}
\label{Fig. 1}
\end{figure}
\end{remark}

\section{Algorithm Performance}

This section states the properties of proposed algorithm \emph{FREE}.

\begin{theorem}
\label{Theo:EnergyLevel} For $\theta \hat{=}\varphi V+O_{{}}^{\max }$ and
\begin{equation}
0<V\leq \frac{S_{{}}^{\max }-w^{\max }-O_{{}}^{\max }}{\varphi },  \label{26}
\end{equation}%
the energy level in the energy storage device satisfies
\begin{equation}
S\left( t\right) \in \left[ 0,S_{{}}^{\max }\right] \text{, }\forall t.
\label{27}
\end{equation}
\end{theorem}

\begin{proof}
The proof can be found in Appendix B.
\end{proof}

The above results demonstrate that although \emph{FREE} solves $\left(
P2\right) $ without considering the constraints\ $\left( \ref{3}\right) $-$%
\left( \ref{5}\right) ,$ the energy level in storage device is feasible in
each time slot.\ The following proposition characterizes the boundedness of
data queue. By recalling $\left( \ref{61}\right) ,$ it is noted that $\theta
$ can be regarded as a charging threshold for the battery and be implemented
by the algorithm. Based on $\left( \ref{26}\right) \,$and $\left( \ref{27}%
\right) ,$ the maximum $V$ is mostly determined by $S_{{}}^{\max },$ which
should be at least larger than the summation of the maximum charging rate $%
w^{\max }$ and the maximum withdrawing rate $O_{{}}^{\max }.$ We can
understand the lower bound of $S_{{}}^{\max }$ by the following observation.
To maintain the feasible value of residual energy in the battery, $%
S_{{}}^{\max }\geq w^{\max }+O_{{}}^{\max }$ is obtained by adding $\left( %
\ref{4}\right) $ and $\left( \ref{5}\right) $ together,\textit{\ i.e.}$%
,w\left( t\right) +O\left( t\right) \leq S^{\max }-S\left( t\right) +S\left(
t\right) =S^{\max },$ $\forall t.$

\begin{proposition}
\label{Prop:boundsQ} Initialize $Q_{n}\left( 0\right) =0,$ $\forall n$. the
data buffer backlogs yield the deterministic bounds $Q_{n}\left( t\right)
\leq Q^{\max }$\ for $\forall t\geq 0,$ $\forall n.$
\end{proposition}

\begin{proof}
The proof can be found in Appendix C.
\end{proof}

In the above results, given$\ Q^{\max }$ the designed \emph{FREE} can
guarantee that the actual queue length is limited by $Q^{\max }$ in each
time slot and the virtual queue also has a deterministic upper bound $%
U^{\max }.$ Theorem \ref{Theo:Nearoptimal} further proves \emph{FREE}
designed for the relaxed problem $\left( P2\right) $ can approach the
optimal solution of original problem$\ \left( P1\right) $ asymptotically.

\begin{theorem}
\label{Theo:Nearoptimal} Given $Q^{\max }>A^{\max 2}+\frac{A^{\max 2}+\mu
^{\max 2}}{2\epsilon },\ $the proposed algorithm \emph{FREE} can achieve the
near optimal performance:
\begin{eqnarray}
&&\underset{\!t\rightarrow \!\infty }{\!\lim \!\inf }\phi \!\!\sum_{n\in
\mathcal{N}}\!\!f\!\!\left( \!\frac{1}{t}\!\sum_{\tau =0}^{t-1}\!\mathbb{E}%
\!\left\{ \!R_{n}\!\left( \!\tau \!\right) \!\right\} \!\right) \!-\!\varphi
\!\frac{1}{t}\!\sum_{\tau =0}^{t-1}\!\mathbb{E}\!\left\{ \!P\!\left( \!\tau
\!\right) \!\right\}  \label{35} \\
&\geq &\mathcal{V}_{1}-\frac{\Xi }{V},  \notag
\end{eqnarray}%
and a bounded virtual queue%
\begin{eqnarray}
&&\underset{t\rightarrow \infty }{\lim \sup }\frac{1}{t}\sum_{\tau
=0}^{t-1}\sum_{n=1}^{N}\mathbb{E}\left\{ U_{n}\left( \tau \right) \right\}
\label{65} \\
&\leq &\frac{\left( \Xi +V\mathcal{V}^{\max }\right) Q^{\max }}{\epsilon
\left( Q^{\max }-A^{\max }\right) -2^{-1}\left( \mu ^{\max 2}+A^{\max
2}\right) }  \notag
\end{eqnarray}%
where $\epsilon $ is a positive variable, $\mathcal{V}_{1}$ is the optimal
objective value of $\left( P1\right) $ and $\Xi $ is defined as
\begin{eqnarray}
\Xi &\hat{=}&\frac{1}{2}NA^{\max }Q^{\max }+N\frac{Q^{\max }-A^{\max }}{%
Q^{\max }}A^{\max 2}  \label{41} \\
&&+\frac{1}{2}\left( w^{\max 2}+O^{\max 2}\right) .  \notag
\end{eqnarray}
\end{theorem}

Before we give the formal proof, the following lemma is needed.

\begin{lemma}
If the packet arrival process $\left\{ A_{n}\left( t\right) ,\forall n\in
\mathcal{N}\right\} ,$ channel state $\mathcal{I}\left( t\right) ,$ and
renewable energy generation$\ w\left( t\right) ,$ are i.i.d., then there
exists a randomized stationary policy $\Psi ^{RS}$ that takes feasible flow
control, radio resource allocation and energy management only based on the\
current system state $\left\{ A_{n}\left( t\right) ,\mathcal{I}\left(
t\right) ,w\left( t\right) \right\} $ at every time slot $t$ while
satisfying the constraints in problem $\left( P2\right) $ and results in the
following conditions:%
\begin{equation*}
\begin{array}{l}
\mathbb{E}\left\{ O^{RS}\left( t\right) \right\} \!=\!\mathbb{E}\left\{
\delta ^{RS}\left( t\right) w\left( t\right) \right\} , \\
\mathbb{E}\left\{ X_{n}^{RS}\left( t\right) \right\} \!=\!\mathbb{E}\left\{
\mu _{n}^{RS}\left( t\right) \right\} =\mathbb{E}\left\{ R_{n}^{RS}\left(
t\right) \right\} \!=\!\kappa _{n},\text{ }\forall \!n\!\in \!\mathcal{N},
\\
\mathbb{E}\left\{ \phi \sum_{n=1}^{N}f\left( R_{n}^{RS}\left( t\right)
\right) \right\} \!-\!\mathbb{E}\left\{ \varphi P^{RS}\left( t\right)
\right\} =\mathcal{V}_{2},%
\end{array}%
\end{equation*}%
where $\kappa _{n}$ belongs to the network region $\Lambda ,$ which is the
set of network capacities under all possible control policies. There also
exists another randomized stationary policy $RS^{\prime }$ that only depends
on the system state and the following conditions hold, $\forall n\in
\mathcal{N}$.
\begin{equation*}
\begin{array}[t]{l}
\mathbb{E}\left\{ O^{RS^{\prime }}\left( t\right) \right\} \!=\!\mathbb{E}%
\left\{ \delta ^{RS^{\prime }}\left( t\right) w\left( t\right) \right\} , \\
\mathbb{E}\left\{ X_{n}^{RS^{\prime }}\left( t\right) \right\} \!=\!\mathbb{E%
}\left\{ \mu _{n}^{RS^{\prime }}\left( t\right) \right\} \!=\!\mathbb{E}%
\left\{ R_{n}^{RS^{\prime }}\left( t\right) \right\} \!=\!\kappa
_{n}\!-\!\epsilon , \\
\mathbb{E}\left\{ \phi \sum_{n=1}^{N}f\left( R_{n}^{RS^{\prime }}\left(
t\right) \right) \right\} \!-\!\mathbb{E}\left\{ \varphi P^{RS^{\prime
}}\left( t\right) \right\} \!=\!\mathcal{V}_{2\epsilon },%
\end{array}%
\end{equation*}%
where $\epsilon \in \Lambda $ is a positive value that can be chosen
arbitrarily close to zero. According to \cite{Neely_book2006}, we have $%
\mathcal{V}_{2\epsilon }\rightarrow \mathcal{V}_{2}$ as $\epsilon
\rightarrow 0.$
\end{lemma}

The proof process of Theorem \ref{Theo:Nearoptimal} is described in the following.
\begin{proof}
The average performance bound $\left( \ref{35}\right) $ is obtained by
comparing the Lyapunov drift of \emph{FREE} with the aforementioned
randomized stationary policy. We use Lyapunov optimization technique to
derive the performance bound. Let $\Theta \left( t\right) =\left( S\left(
t\right) ,\mathbf{U}\left( t\right) ,\mathbf{Q}\left( t\right) \right) ,$
where $\mathbf{U}\left( t\right) =\left( U_{n}^{{}}\left( t\right) ,n\in
\mathcal{N}\right) ,$ $\mathbf{Q}\left( t\right) =\left( Q_{n}^{{}}\left(
t\right) ,n\in \mathcal{N}\right) .$\ Define the Lyapunov function as
\begin{eqnarray}
&&\mathcal{L}\!\left( \!\Theta \!\left( t\right) \!\right)  \label{57} \\
&\!\!=\!\!&\!\frac{1}{2}\!\!\sum_{n\in \mathcal{N}}\!\!\left[ \!\frac{%
Q^{\max }\!-\!A^{\max }}{Q^{\max }}U_{n}^{2}\!\left( t\right) \!+\!\frac{%
U_{n}^{{}}\!\left( t\right) \!Q_{n}^{2}\!\left( t\right) \!}{Q^{\max }}%
\right] \!\!+\!\!\frac{1}{2}\!\left( \!S\!\left( t\right) \!\!-\!\!\theta
\right) ^{2}  \notag
\end{eqnarray}

\begin{figure*}[tbp]
\centering
\begin{minipage}{.32\textwidth}
\centering
\includegraphics[width=6cm]{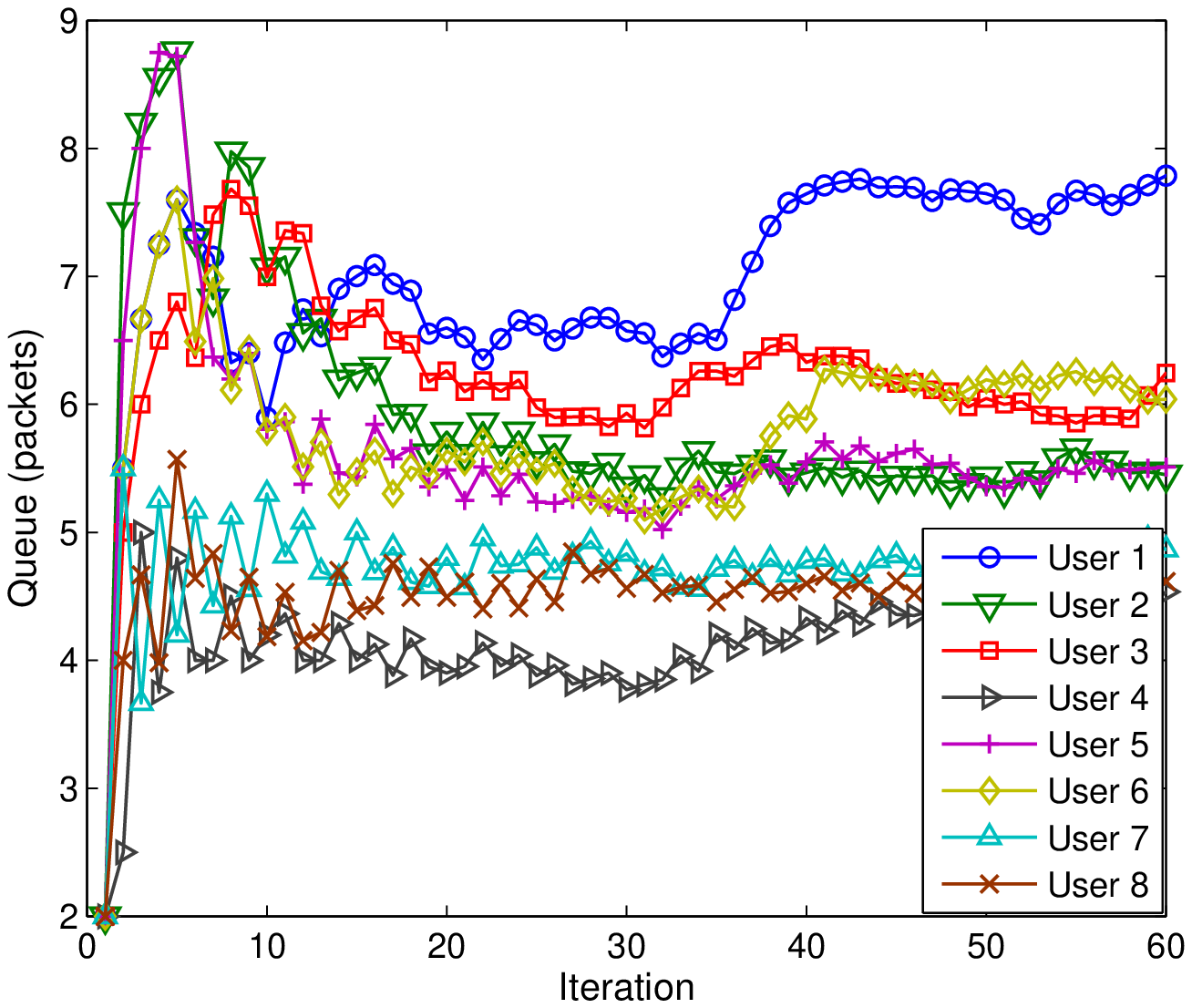}
\caption{Dynamics of actual\ queue $\left\{ Q_{n}\left( t\right) \right\} $.}
\label{Fig. 3}
\end{minipage}%
\begin{minipage}{.32\textwidth}
\centering
\includegraphics[width=5.8cm]{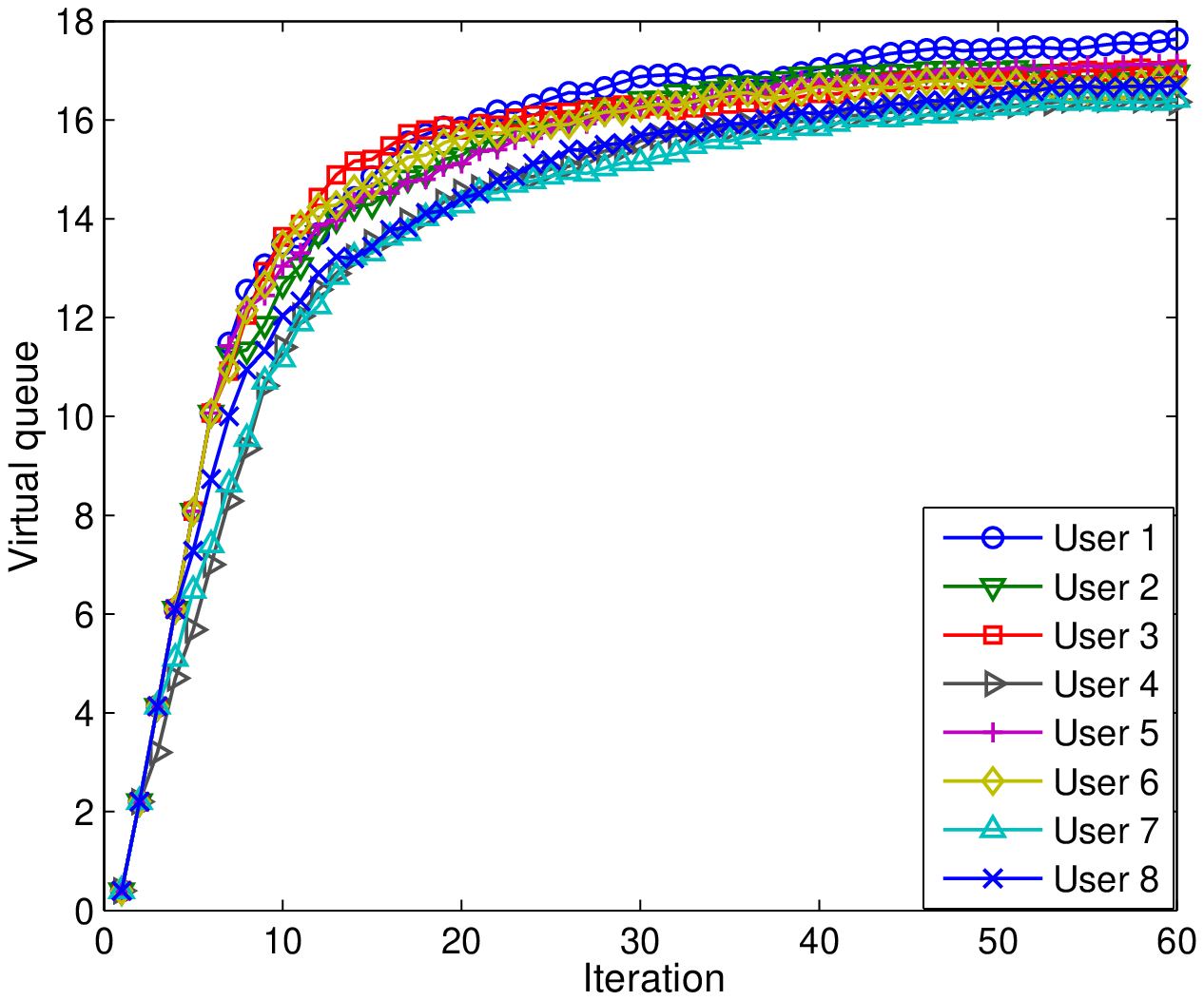}
\caption{Dynamics of virtual\ queue $\left\{ U_{n}\left( t\right) \right\} $.}
\label{Fig. 4}
\end{minipage}%
\begin{minipage}{.32\textwidth}
\centering
\includegraphics[width=5.8cm]{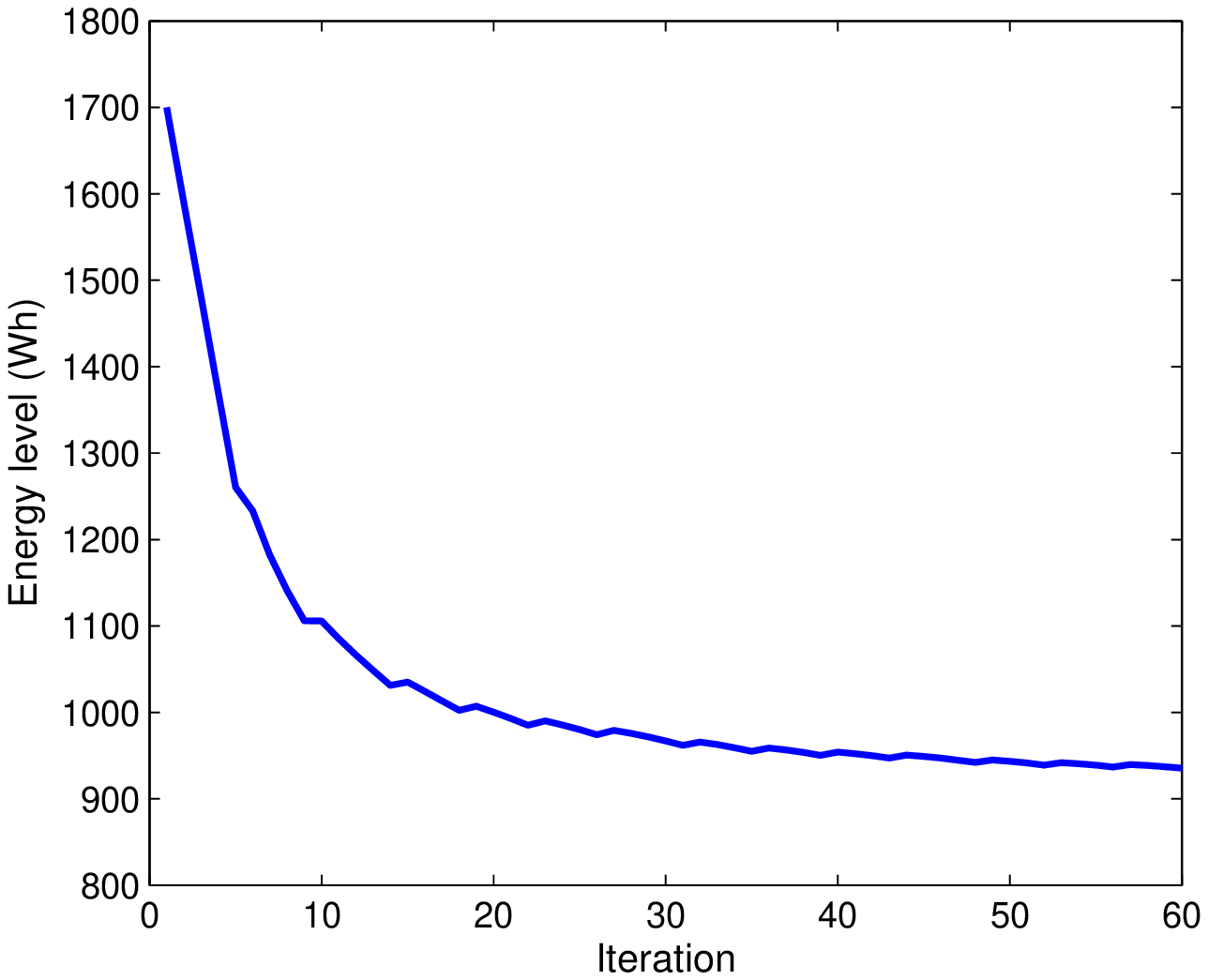}
\caption{Energy level $S\left( t\right) $\ in the storage device.}
\label{Fig. 5}
\end{minipage}
\end{figure*}

Define the conditional Lyapunov drift as follows,
\begin{equation}
\Delta \mathcal{L}\left( \Theta \left( t\right) \right) \hat{=}\mathbb{E}%
\left\{ \mathcal{L}\left( \Theta \left( t+1\right) \right) -\mathcal{L}%
\left( \Theta \left( t\right) \right) \left\vert \Theta \left( t\right)
\right. \right\} ,  \label{40}
\end{equation}%
where the conditional expectation is taken over the system queue. The drift
of the first two items in $\left( \ref{57}\right) $ can be found in$\ $\cite%
{Xingzheng_TVT}. According to the battery dynamics in $\left( \ref{3}\right)
,$ the drift of the third item in $\left( \ref{57}\right) $ is%
\begin{eqnarray}
&&\frac{1}{2}\left( \left( S\left( t+1\right) -\theta \right) ^{2}-\left(
S\left( t\right) -\theta \right) ^{2}\right)  \label{39} \\
&=&\frac{1}{2}\left( S_{{}}^{2}\left( t+1\right) -S_{{}}^{2}\left( t\right)
-2\theta \left( S\left( t+1\right) -S\left( t\right) \right) \right)  \notag
\\
&=&\frac{1}{2}(\left( S\left( t\right) -O\left( t\right) +\delta \left(
t\right) w\left( t\right) \right) ^{2}-S_{{}}^{2}\left( t\right)  \notag \\
&&-2\theta _{{}}\left( S_{{}}\left( t+1\right) -S\left( t\right) \right) )
\notag \\
&\leq &\frac{1}{2}(w^{\max 2}+O^{\max 2}  \notag \\
&&-2\left( S\left( t\right) -\theta \right) \left( O\left( t\right) -\delta
\left( t\right) w\left( t\right) \right) )  \notag \\
&=&\frac{1}{2}\left( w^{\max 2}+O^{\max 2}\right)  \notag \\
&&-\left( S\left( t\right) -\theta \right) \left( p_{B}\left( t\right)
+\Delta p_{B}\left( t\right) -J\left( t\right) -\delta _{{}}\left( t\right)
w_{{}}\left( t\right) \right)  \notag
\end{eqnarray}%
By subtracting $V\mathbb{E}\left\{ \left( \sum_{n\in \mathcal{N}}^{{}}\phi
f\left( X_{n}\left( t\right) \right) -\varphi P\left( t\right) \right)
\left\vert \Theta \left( t\right) \right. \right\} $ from the conditional
drift of $\left( \ref{57}\right) $, we have%
\begin{eqnarray}
&&\Delta \!\mathcal{L}\!\left( \!\Theta \!\left( t\right) \!\right)
\!\!-\!\!V\mathbb{E}\!\!\left\{ \!\!\left( \!\sum_{n\in \mathcal{N}%
}^{{}}\!\phi f\!\left( \!X_{n}\!\left( t\right) \!\right) \!-\!\varphi
\!P\left( t\right) \!\right) \left\vert \!\Theta \!\left( t\right) \!\right.
\!\right\}  \label{43} \\
&\!\leq \!&\Xi \!+\!\sum_{n\in \mathcal{N}}^{{}}\!\mathbb{E}\!\left\{ \!\Phi
_{1n}\!\left( t\right) \!+\!\Phi _{2n}\!\left( t\right) \!-\!\Phi
_{3n}\!\left( t\right) \!+\!\Phi _{4}\!\left( t\right) \!\left\vert \!\Theta
\!\left( t\right) \!\right. \!\right\}  \notag \\
&&\!-\!\left( S\!\left( t\right) \!-\!\theta \right) \!\Delta p_{B}\!\left(
t\right) \!+\!\mathbb{E}\!\left\{ \!\varphi V\!\sum_{i\in \mathcal{R}%
}\!\Delta p_{i}\!\left( t\right) \!\left\vert \!\Theta \!\left( t\right)
\!\right. \!\right\}  \notag \\
&&\!+\!\frac{1}{2}\!\sum_{n\in \mathcal{N}}^{{}}\!\frac{U_{n}\!\left(
t\right) \!}{Q^{\max }}\!\left( \mu ^{\max 2}\!+\!A^{\max 2}\right)  \notag
\end{eqnarray}%
where $\Xi $ is defined in $\left( \ref{41}\right) $ and $\mathcal{R}\left(
n\right) $ is the set of RSs that assist the transmissions towards user $n$,%
\begin{equation*}
\Phi _{1n}\left( t\right) \hat{=}\frac{Q^{\max }-A^{\max }}{Q^{\max }}%
U_{n}\left( t\right) X_{n}\left( t\right) -V\phi f\left( X_{n}\left(
t\right) \right)
\end{equation*}%
\begin{equation*}
\Phi _{2n}\left( t\right) \hat{=}R_{n}\left( t\right) \frac{U_{n}\left(
t\right) }{Q^{\max }}\left( Q_{n}\left( t\right) -\left( Q^{\max }-A^{\max
}\right) \right)
\end{equation*}%
\begin{equation*}
\Phi _{3n}\left( t\right) \!\hat{=}\!\frac{U_{n}\!\left( t\right)
\!Q_{n}\!\left( t\right) \!\mu _{n}\!\left( t\right) \!}{Q^{\max }}%
\!+\!\left( \!S\!\left( t\right) \!-\!\theta \right) \!p_{B}\!\left(
t\right) \!-\!V\!\varphi \!\!\sum_{i\in \mathcal{R}\!\left( n\right)
\!}\!p_{i}\!\left( t\right)
\end{equation*}%
\begin{equation*}
\Phi _{4}\left( t\right) \hat{=}\left( \varphi V+S\left( t\right) -\theta
\right) J\left( t\right) +\left( S\left( t\right) -\theta \right) \delta
\left( t\right) w\left( t\right)
\end{equation*}%
It can be found that the subproblems $\left( \ref{60}\right) ,$ $\left( \ref%
{59}\right) ,$ $\left( \ref{23}\right) $ and $\left( \ref{24}\right) $ in
\emph{FREE }minimizes $\sum_{n\in \mathcal{N}}^{{}}\mathbb{E}\left\{ \Phi
_{1n}\left( t\right) +\Phi _{2n}\left( t\right) -\Phi _{3n}\left( t\right)
+\Phi _{4}\left( t\right) \left\vert \Theta \left( t\right) \right. \right\}
$ on the right hand side of $\left( \ref{43}\right) $ over all possible
solutions including the randomized stationary policy. Substituting$\ $the
randomized stationary policy $RS^{\prime }$ into $\Phi _{1n}\left( t\right) $
and substituting policy $RS$ into $\Phi _{2n}\left( t\right) ,\Phi
_{3n}\left( t\right) $ and $\Phi _{4}\left( t\right) $ on the right hand
side of $\left( \ref{43}\right) $ it is obtained that,
\begin{eqnarray*}
&&\Delta \!\mathcal{L}\!\left( \!\Theta \!\left( t\right) \!\right) \!-\!V\!%
\mathbb{E}\!\left\{ \!\left( \!\!\sum_{n\in \mathcal{N}}^{{}}\!\phi
f\!\left( X_{n}\!\left( t\right) \right) \!-\!\varphi P\!\left( t\right)
\!\right) \left\vert \Theta \!\left( t\right) \!\right. \!\right\} \\
&\!\leq \!&\Xi \!-\!\frac{\epsilon }{2}\!\sum_{n\in \mathcal{N}}^{{}}\!\frac{%
U_{n}\!\left( t\right) \!}{Q^{\max }}(2\!\left( Q^{\max }\!-\!A^{\max
}\right) \\
&&-\epsilon ^{-1}\!\left( \mu ^{\max 2}\!+\!A^{\max 2}\right) )\!-\!V%
\mathcal{V}_{2\epsilon }
\end{eqnarray*}%
If $Q^{\max }\geq \frac{\mu ^{\max 2}+A^{\max 2}}{2\epsilon }+A^{\max },$ by
Theorem 5.4 in \cite{Neely_book2006}, it is obtained that $\left( \ref{65}%
\right) $ and%
\begin{equation}
\underset{t\rightarrow \infty }{\lim \inf }\frac{1}{t}\!\sum_{\tau
=0}^{t-1}\!\mathbb{E}\!\left\{ \!\sum_{n=1}^{N}\!\phi f\!\left(
X_{n}\!\left( \tau \right) \!\right) \!-\!\varphi P\!\left( \tau \right)
\!\right\} \!\geq \!\mathcal{V}_{2\epsilon }\!-\!\frac{\Xi }{V}  \label{63}
\end{equation}%
hold. Since the virtual\ queue $U_{n}\left( t\right) $ is stable by $\left( %
\ref{65}\right) $\emph{,} we have $\lim_{t\rightarrow \infty }\frac{1}{t}%
\sum\limits_{\tau =0}^{t-1}\mathbb{E}\{R_{n}\left( \tau \right) \}\geq
\lim_{t\rightarrow \infty }\frac{1}{t}\mathbb{E}\{X_{n}\left( \tau \right)
\}.$ Hence we have
\begin{equation}
f\left( \bar{R}_{n}\right) \geq f\left( \bar{X}_{n}\right) \geq \overline{%
f\left( X_{n}\right) },  \label{64}
\end{equation}%
where the first inequality is due to the\ non-decreasing property of $%
f\!\left( \cdot \right) \!$ and the second inequality is Jensen's
inequality. Note that $\left( \ref{63}\right) $ holds for all $\epsilon .$
Due to $\mathcal{V}_{2\epsilon }\!\!\rightarrow \!\!\mathcal{V}_{2}$ as $%
\epsilon \!\!\rightarrow \!\!0,$ and $\mathcal{V}_{2}\!>\!\mathcal{V}_{1}$
we have $\left( \ref{35}\right) $ in \emph{Theorem \ref{Theo:Nearoptimal} }%
by substituting $\left( \ref{64}\right) $ into $\left( \ref{63}\right) $.
\end{proof}

\begin{remark}
Although the designed per-slot sub-problems $\left( \ref{59}\right) ,$ $%
\left( \ref{60}\right) \,,$ $\left( \ref{23}\right) ,$ $\left( \ref{24}%
\right) $ without considering the future network performance, jointly solve $%
\left( P2\right) $ instead of $\left( P1\right) ,$ the resulted solution can
approach the optimal objective value of $\left( P1\right) $ arbitrarily by
regulating $V$. This result is reasonable since a larger storage (a larger $%
V $) can store more renewable energy to save on-grid energy to achieve a
larger objective value.
\end{remark}

\begin{remark}
The result in \emph{Theorem \ref{Theo:Nearoptimal}} is based on the
assumption of i.i.d stochastic process. It can be generalized to non i.i.d.
Markov modulated scenarios using the techniques developed in \cite%
{Neely_book2006} \cite{Neely_book_2010}.
\end{remark}

\section{Simulation Results}

This section presents the simulation results for the proposed algorithm
\emph{FREE.} A cell with 2 km radius is considered, where eight\ mobile
users are uniformly distributed in the cell and four RSs are assumed to be
uniformly placed in the midway of BS and cell boundary. The simulation
parameters are set as follows except for other specification. In every time
slot, the BS receives new packets with destinations of the corresponding
users according to an i.i.d Poisson arrival process with average rate of $8$
packets/s. The packet size has an\ exponential distribution with mean packet
size of $5000$ bits/packet.\ The buffer size for each user is $10$ packets.
The total bandwidth is $10$ MHz with $128$ subcarriers. The channel gain of
any transmission pair in networks consists of a small-scale Rayleigh fading
component and a large-scale path loss component with path loss factor of $4.$
The noise level is assumed to be $10^{-10}$ W and the gap to capacity is set
to be $1$. The static power consumption of the BS is $194.24$ W and its
maximum transmission power is $20$ W without particular specification. For
any RS, its static power consumption is $40$ W and the maximum transmission
power is $10$ W. The power mask on each subcarrier is $\hat{P}\left(
m\right) =0.2$ W. The capacity of the BS\ storage device is $3000$ Wh unless
otherwise specified. The maximum discharge rate\ of the storage device is $%
O^{\max }=1.5\times \left( p_{B}^{\max }+\Delta p_{B}\right) $ W. The
maximum power from the grid to supply the BS directly is $J^{\max }=O^{\max
}\,$. The average\ generated solar power has two states corresponding to
sunny day and cloudy day: $195$ Wh with probability $0.6$ and $100$ Wh$\,$%
with$\ $probability $0.4$, respectively \cite{duisit_niyato}. The parameter $%
\theta \ $is set as $\varphi V+O_{{}}^{\max }.$ The utility function is
chosen as $f\left( x\right) =\log \left( 1+x\right) .$



Firstly, we verify the stability of proposed algorithms. From Fig. \ref{Fig.
3} to Fig. \ref{Fig. 5}, it is observed that all the actual queues $\left\{
Q_{n}\left( t\right) \right\} $, virtual queues $\left\{ U_{n}\left(
t\right) \right\} $ and energy level $\left\{ S\left( t\right) \right\} $ in
the storage device have the trend of stability. Especially Fig. \ref{Fig. 3}
shows that all actual queues are smaller than the upper bound of $10$
packets and Fig. \ref{Fig. 4} demonstrates that all virtual queues are
bounded, which verifies the results in \emph{Proposition \ref{Prop:boundsQ}}
and $\left( \ref{65}\right) $\emph{.} Fig. \ref{Fig. 5} demonstrates the
results in \emph{Theorem \ref{Theo:EnergyLevel}} that the\ energy level in
the storage device of the BS is always nonnegative and below the capacity of
storage device.

\begin{figure}[tbph]
\center%
\includegraphics[width=9cm,
height=5cm]{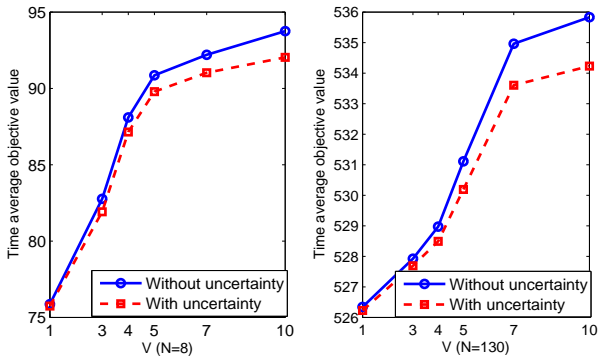} 
\caption{Time average objective value with different $V.$ Other parameters
include $\protect\varphi =0.5$ and $\protect\phi =16.$}
\label{Fig. 6}
\end{figure}

The simulations are run for $1.2\times 10^{4}$ iterations\ and the results
in Fig. \ref{Fig. 6} to Fig. \ref{Fig. 10} are obtained by averaging the
last $5000$ iterations. It is shown that the time-average\ tradeoff value in
$\left( P1\right) $ grows with the increase of $V$, which verifies $\left( %
\ref{35}\right) .$ It is further noticed that with $N=130$ the time average
objective value increases slower than that with $N=8$, since a larger $N$ in
$\Xi $ $\left( \ref{41}\right) $ brings a larger gap between the optimal
objective value and actual value. To decrease the gap, a larger $V$ is
needed. Fig. \ref{Fig. 6} also demonstrates that even with the same $N$ the
system objective value degrade a little when all the channel state is
measured with 30\% uncertainty, \textit{i.e. }$\hat{H}\left( t\right)
=H\left( t\right) +\Delta H\left( t\right) ,$ where $H\left( t\right) $ is
the real channel gain, $\left\vert \Delta H\left( t\right) /H\left( t\right)
\right\vert =30\%$ is channel uncertainty and $\hat{H}\left( t\right) $ the
measured channel gain.

\begin{figure}[tbph]
\center\includegraphics[width=8cm, height=5.8cm]{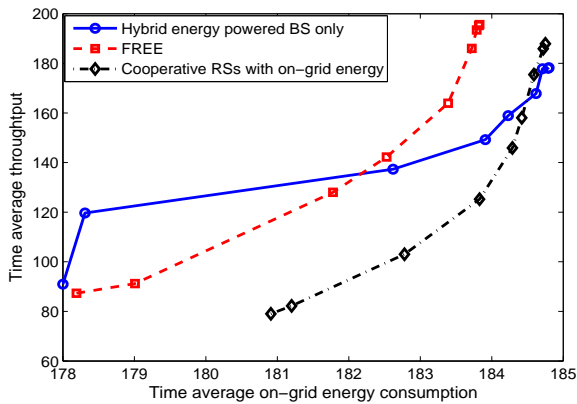}
\caption{The time average throughput vs time average on-grid energy
consumption.}
\label{Fig. 7}
\end{figure}

The tradeoff between the system throughput and on-grid energy consumption is
shown in Fig. \ref{Fig. 7}, which is depicted by decreasing $\varphi $
gradually while $\phi $ and $V$ are fixed. In order to measure the effects
of renewable energy and cooperative relay, \emph{FREE }is compared with the
other two schemes: the scheme with hybrid energy powered BS without relay,
and the scheme with both RSs and BS powered by on-grid energy alone. The
compared schemes are also designed by Lyapunov optimization. In Fig. \ref%
{Fig. 7} it is found that for on-grid energy consumption smaller than 182\
the system throughput of \emph{FREE }is even smaller than the scheme with
hybrid energy powered\ BS only. This range corresponds to a larger $\varphi
, $ which forces \emph{FREE} to give up using most RSs to save on-grid
energy. However the RSs are still there generating constant power
consumption $\Delta p_{i}$ without any throughput contribution. As $\varphi $
continues to decrease, more RSs are employed by \emph{FREE }to obtain more
throughput than the scheme with hybrid energy powered\ BS only. It is also
found that mostly the scheme with cooperative RSs powered by on-grid energy
alone is not energy-efficient compared with the other two schemes. When more
on-grid energy is consumed, its throughput is larger than the scheme with BS
only but still smaller than \emph{FREE. }Thus, we can see that the benefits
due to renewable energy and cooperative relay depend on proper parameter
setting.

In Fig. \ref{Fig. 9} we compare \emph{FREE }with the other\ two related
works in \cite{Weiyu_2007}\ and \cite{WangTao_TVT}. Lagrange dual
decomposition method and discrete particle swarm optimization method are
adopted respectively by \cite{Weiyu_2007} and \cite{WangTao_TVT} for
resource allocation in cooperative relay networks, \textit{i.e.}, to pursue
\begin{equation*}
\begin{array}{cc}
\max & \phi \sum_{n}f\left( \mu _{n}\left( t\right) \right) -\varphi P\left(
t\right)%
\end{array}%
\end{equation*}%
For fair comparison, we assume that the BSs in \cite{Weiyu_2007} and \cite%
{WangTao_TVT} are also powered by hybrid energy. We assume there is a
heuristic renewable energy management scheme for \cite{Weiyu_2007} and \cite%
{WangTao_TVT}. With the heuristic scheme, the harvested energy is always
charged into the battery as long as there is enough space, and the BS always
uses the saved renewable energy with priority. Since the two related works
do not consider flow control, we measure the time average\ service rate and
the time average on-grid energy consumption with the decreasing $\varphi .$
It is found in Fig. \ref{Fig. 9} that the performance of \emph{FREE}
outperforms the other two schemes, since it adapts the resource allocation
to the arrival packets and battery states. Moreover, the performance of \cite%
{WangTao_TVT} is better than \cite{Weiyu_2007}, since the former takes into
account dynamic channel state explicitly.

\begin{figure}[tbph]
\center\includegraphics[width=9cm, height=5.8cm]{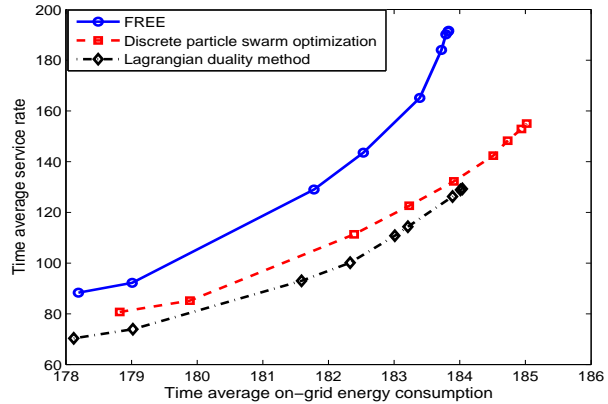} 
\caption{Performance comparison between the proposed algorithm and the other
algorithms}
\label{Fig. 9}
\end{figure}

In Fig. \ref{Fig. 10}, we use the fairness index $FI=\left( \sum_{n=1}^{N}%
\bar{R}_{n}\right) ^{2}/\left( N\sum_{n=1}^{N}\bar{R}_{n}^{2}\right) $ to
measure the fairness. If all users have the same throughput, $FI$ is $1.$ It
is found that with the number of user from 10 to 60, the proposed scheme
maintains better fairness compared with the throughput maximization scheme.
\begin{figure}[tbph]
\center\includegraphics[width=8.5cm, height=5.8cm]{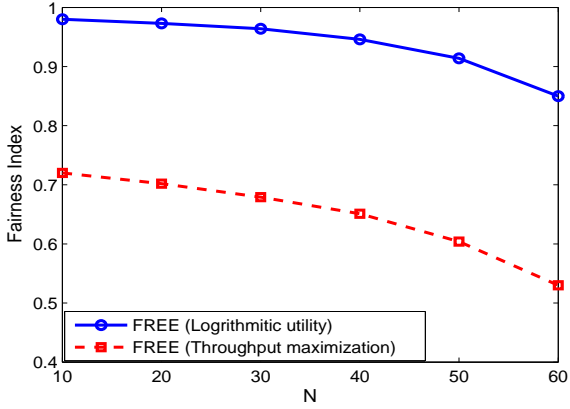}
\caption{Fairness index}
\label{Fig. 10}
\end{figure}

\section{Conclusion}

In this paper, we consider to optimize the tradeoff between downlink
throughput utility and on-grid energy consumption of an OFDMA cellular
network with the assistance of multiple RSs. The BS is powered by the
conventional utility grid and the\ renewable energy. A joint design of flow
control, radio resource allocation and energy management is proposed to
handle the coupling between the energy consumption of the cooperative
network and the storage energy allocation. The joint design scheme, namely
\emph{FREE} can be implemented by exploiting the local channel states,
buffer states and energy states. Technical proof is established to ensure
that \emph{FREE} can produce a close-to-optimal performance while allowing
finite data buffer and energy storage device. Simulation results show that
\emph{FREE }demonstrates better performance compared with other schemes.

\begin{center}
\textsc{Appendix A}
\end{center}

\textsc{1. Solution to }$\left( \ref{11}\right) $

We first pursue the solution to $\left( \ref{11}\right) .$ To address the
variable coupling in the objective function, we let
\begin{equation}
d_{i,n}^{m}\hat{=}p_{B}^{m}H_{B,n}^{m}+p_{i,n}^{m}H_{i,n}^{m}.  \label{50}
\end{equation}%
The problem $\left( \ref{11}\right) $ is rewritten as
\begin{equation}
\begin{array}{cc}
\max &
\begin{array}{c}
\Upsilon _{1}=q_{n}\left( t\right) \log \left( 1+d_{i,n}^{m}\right) \\
\!+\!\left(\!S\!\left( t\right)\!-\!\theta\!-\!\lambda _{B}\!+\!\frac{%
\!\left( V\varphi \!+\!\lambda _{i}\!\right)\!H_{B,n}^{m}}{H_{i,n}^{m}}%
\!\right) p_{B}^{m}\!-\!\frac{V\varphi\!+\!\lambda _{i}}{H_{i,n}^{m}}%
d_{i,n}^{m}%
\end{array}
\\
\text{s.t.} & 0\!\leq\!d_{i,n}^{m}\!\leq\!\min\!\left\{%
\!p_{B}^{m}H_{B,i}^{m},\bar{P}\!\left(
m\right)\!H_{i,n}^{m}\!+\!p_{B}^{m}H_{B,n}^{m}\!\right\} \\
& 0\leq p_{B}^{m}\leq \bar{P}\left( m\right) ,%
\end{array}%
,  \label{13}
\end{equation}

If $S\left( t\right) -\theta -\lambda _{B}+\frac{\left( V\varphi +\lambda
_{i}\right) H_{B,n}^{m}}{H_{i,n}^{m}}\geq 0$, then the optimal transmission
power of the BS over subcarrier $m$ is $p_{B}^{m}=\bar{P}\left( m\right) .$
Substituting $p_{B}^{m}=\bar{P}\left( m\right) $ into $\left( \ref{13}%
\right) ,$ the optimal $d_{i,n}^{m}$ is
\begin{equation*}
d_{i,n}^{m}\!=\!\max\!\left\{\!0,\!\min\!\left\{\!\tilde{d}_{i,n}^{m},\bar{P}%
\!\left( m\right)\!H_{B,i}^{m},\bar{P}\!\left( m\right)\!\left(
H_{i,n}^{m}\!+\!H_{B,n}^{m}\!\right)\!\right\}\!\right\}
\end{equation*}%
where $\tilde{d}_{i,n}^{m}=\frac{q_{n}\left( t\right) H_{i,n}^{m}}{V\varphi
+\lambda _{i}}-1$ is the solution to $\nabla _{d_{i,n}^{m}}\Upsilon _{1}=0.$
Then based on the definition of $d_{i,n}^{m}$ in $\left( \ref{50}\right) $
the optimal $p_{i,n}^{m}$ is $\left( \frac{d_{i,n}^{m}-\bar{P}\left(
m\right) H_{B,n}^{m}}{H_{i,n}^{m}}\right) ^{+}.$

If $S\left( t\right) -\theta -\lambda _{B}+\frac{\left( V\varphi +\lambda
_{i}\right) H_{B,n}^{m}}{H_{i,n}^{m}}<0,$ then the optimal solutions to $%
\left( \ref{13}\right) $ are $p_{B}^{m}=0$ and $d_{i,n}^{m}=0$ and then $%
p_{i,n}^{m}=0$

\textsc{2. Solution to }$\left( \ref{12}\right) $

Let the objective of $\left( \ref{12}\right) $ be $\Upsilon _{2}.$ Since the
objective function of $\left( \ref{12}\right) $ is decreasing with respect
to $p_{i,n}^{m},$ the optimal value is $p_{i,n}^{m}=0.$ There are two cases
to be considered for the optimal $p_{B}^{m}$. If $H_{B,i}^{m}>H_{B,n}^{m},$
then the optimal transmission power of the BS is $p_{B}^{m}=0.$ If $%
H_{B,i}^{m}\leq H_{B,n}^{m},$ the optimal value is
\begin{equation*}
p_{B}^{m}=\left\{
\begin{array}{cc}
\bar{P}\left( m\right) & \text{if }\!S\!\left(
t\right)\!-\!\theta\!-\!\lambda _{B}\!\geq\!0 \\
\!\left(\!\frac{q_{n}\!\left( t\right)\!}{\lambda _{B}\!-\!S\!\left(
t\right)\!+\!\theta }\!-\!\frac{1}{H_{B,i}^{m}}\!\right) _{0}^{\bar{P}%
\!\left( m\right)\!} & \text{otherwise}%
\end{array}%
\right. ,
\end{equation*}
where the\ optimal $p_{B}^{m}$ in the case of $S\left( t\right) -\theta
-\lambda _{B}<0$ is obtained by solving $\nabla _{p_{B}^{m}}\Upsilon _{2}=0$
considering the power mask and non-negative constraints.

\begin{center}
\textsc{Appendix B}

\textsc{Proof of Theorem \ref{Theo:EnergyLevel}}
\end{center}

We now use the induction method to derive $\left( \ref{27}\right) $.
Initially, without any action on the storage device, $0\leq S\left( 0\right)
\leq S_{{}}^{\max }$ holds. Supposing $\left( \ref{27}\right) $ holds for
time slot $t$, we will prove it holds for slot $t+1$. There are the
following three cases to be considered.

i) If $0\leq S\left( t\right) <O_{{}}^{\max }$, then
\begin{eqnarray}
&&S\!\left( t\!+\!1\right)\!  \notag \\
&=&S\!\left( t\right)\!+\!\delta _{{}}^{\ast }w\!\left(
t\right)\!-\!O_{{}}^{\ast }\!\left( t\right)\!  \notag \\
&=&S\!\left( t\right)\!+\!w_{{}}\!\left( t\right)\!-\!\max\!\left\{
0,p_{B}^{\ast }\!\left( t\right)\!+\!\Delta p_{B}^{{}}\!\left(
t\right)\!-\!J^{\ast }\!\right\}  \label{28} \\
&=&S\!\left( t\right)\!+\!w_{{}}\!\left( t\right)\! \geq\! 0  \notag
\end{eqnarray}%
where $\delta _{{}}^{\ast }$ and $O_{{}}^{\ast }\left( t\right) $ are
derived from solving $\left( \ref{24}\right) .$\ Due to $\theta \hat{=}%
\varphi V+O_{{}}^{\max },$ it falls into the second situation of the
solution to $\left( \ref{24}\right) .$

Next we show that $S\left( t+1\right) \leq S_{{}}^{\max }$. By $\left( \ref%
{3}\right) ,$ we have%
\begin{equation}
S\!\left( t\!+\!1\right) \!<\!S\!\left( t\right) \!+\!w_{{}}^{\max
}\!<\!O_{{}}^{\max }\!+\!w_{{}}^{\max }  \label{30}
\end{equation}%
By substituting the upper bound of $V$ $\left( \ref{26}\right) $ into $%
\left( \ref{30}\right) ,$ we have
\begin{equation*}
S\left( t+1\right) <S_{{}}^{\max }-V\varphi <S_{{}}^{\max }
\end{equation*}

ii) If $O_{{}}^{\max }\leq S\left( t\right) <S_{{}}^{\max }-w^{\max },$ we
have
\begin{equation*}
S\left( t+1\right) =S\left( t\right) +\delta _{{}}^{\ast }w\left( t\right)
-O^{\ast }\left( t\right) \geq O_{{}}^{\max }-O^{\ast }\left( t\right) \geq
0.
\end{equation*}%
Moreover,
\begin{eqnarray*}
S\left( t+1\right) &\leq &S_{{}}\left( t\right) +w^{\max } \\
&<&S_{{}}^{\max }-w^{\max }+w^{\max }=S_{{}}^{\max }
\end{eqnarray*}

iii) If $S_{{}}^{\max }-w^{\max }\leq S\left( t\right) \leq S^{\max },$ it
means $S\left( t\right) \geq \theta ,$ which is obtained by substituting the
upper bound of $V$ $\left( \ref{26}\right) $ into $\theta \hat{=}\varphi
V+O_{{}}^{\max }.$ Thus, we have $\delta _{{}}^{\ast }=0\ $and
\begin{eqnarray}
S\left( t+1\right) &=&S\left( t\right) +\delta _{{}}^{\ast }w\left( t\right)
-O_{{}}^{\ast }\left( t\right)  \notag \\
&>&S_{{}}^{\max }-w^{\max }-O^{\max }>0.  \label{51}
\end{eqnarray}%
\begin{eqnarray*}
S\left( t+1\right) &=&S\left( t\right) +\delta _{{}}^{\ast }w\left( t\right)
-O_{{}}^{\ast }\left( t\right) \\
&=&S\left( t\right) -O_{{}}^{\ast }\left( t\right) <S^{\max }.
\end{eqnarray*}

To summarize, $S\!\left( t\!+\!1\right)\!\in\!\left[0,S^{\max }\right]\! $
holds if $S\!\left( t\right)\!\in\!\left[0,S^{\max }\right]\! .$

\begin{center}
\textsc{Appendix C}

\textsc{Proof of Proposition \ref{Prop:boundsQ}}
\end{center}

$Q_{n}\left( 0\right) =0<Q^{\max }.$ Supposing that $Q_{n}\left( t\right)
\leq Q_{{}}^{\max }$, we have two cases. $(a)$ $Q_{n}\left( t\right) \leq
Q_{{}}^{\max }-A^{\max }.$ Obviously, $Q_{n}\left( t+1\right) \leq
Q_{n}\left( t\right) +A^{\max }\leq Q_{{}}^{\max }$ according to the
dynamics of $\left( \ref{33}\right) .$ $(b)$ If $Q^{\max }\geq Q_{n}\left(
t\right) >Q_{{}}^{\max }-A^{\max },$ then $R_{n}\left( t\right) =0$
according to $\left( \ref{32}\right) $. Thus, we have $Q_{n}\left(
t+1\right) \leq Q_{n}\left( t\right) \leq Q^{\max }$.

%

\end{document}